\pdfoutput=1
\documentclass[]{l4dc2026}

\usepackage{amsfonts}  % Additional math fonts
\usepackage{mathtools} % Extends amsmath
\usepackage{commath}   % Enhanced math commands
\usepackage{cancel}    % For strike-through math
\usepackage{bm}        % Bold math symbols
\usepackage{xcolor}    % For colored text
\usepackage{microtype} % Better typography
\usepackage{etoolbox}  % Toolbox for LaTeX programming
\usepackage{comment}   % For block comments
\usepackage{mdframed}  % For framed environments
\usepackage{xspace}    % For better spacing after macros
\usepackage{booktabs}  % For better looking tables
\usepackage{siunitx}   % For SI units and numbers
\usepackage{float}     % For better float placement
\newcommand{\neuralUQ}{\cite{neural-kalman-anonymous}\xspace}

\usepackage{doi}          % For proper DOI formatting

\newtheorem{problem}{Problem}
\DeclareMathOperator{\trace}{\operatorname{tr}}
\DeclareMathOperator{\expect}{\mathbb{E}}
\DeclareMathOperator{\probability}{\mathbb{P}}

\DeclareMathOperator{\Cov}{\operatorname{Cov}}

\DeclareMathOperator{\normal}{\mathrm N}
\DeclareMathOperator{\Normal}{\mathrm N^*}

\title[Filtering and smoothing neural networks]{Assumed Density Filtering and Smoothing with Neural Network Surrogate Models}
\usepackage{times}
\coltauthor{%
 \Name{Simon Kuang} \Email{slku@ucdavis.edu}
 \AND
 \Name{Xinfan Lin} \Email{lxflin@ucdavis.edu}\\
 \addr University of California, Davis%
}

\renewcommand{\cbr}[1]{\ensuremath{\left\{#1\right\}}}

\begin{document}

\maketitle

\begin{abstract}%
 The Kalman filter and Rauch-Tung-Striebel (RTS) smoother are optimal for state estimation in linear dynamic systems.
With nonlinear systems, the challenge consists in how to propagate uncertainty through the state transitions and output function.
For the case of a neural network model, we enable accurate uncertainty propagation using a recent state-of-the-art analytic formula for computing the mean and covariance of a deep neural network with Gaussian input.
We argue that cross entropy is a more appropriate performance metric than RMSE for evaluating the accuracy of filters and smoothers.
We demonstrate the superiority of our method for state estimation on a stochastic Lorenz system and a Wiener system, and find that our method enables more optimal linear quadratic regulation when the state estimate is used for feedback.
Code available at \url{https://github.com/simontheflutist/analytic-moments}.
\end{abstract}


\section{Introduction}
Neural networks can model dynamic systems by virtue of their universal approximator property \citep{narendra_neural_1992,masri_identification_1993}.
They can be supervised to fit trajectory data \citep{pillonetto_deep_2025}
or ``physics-informed'' to satisfy known equations of motion \citep{mohajerin_multistep_2019,michalowska_neural_2024}.
Beyond function approximation, there is deeper interest in analyzing what neural networks have to offer to the control theorist and vice versa, by finding analysis pathways in the compositional structure of an artificial neural network \citep{junnarkar_synthesizing_2024,xu_eclipse_2024}.
% At the intersection of neural networks and dynamics, there is interest in going beyond function approximation to find ways of working

% \textbf{Premise 1.} \emph{Neural networks are effective for modeling dynamic systems.}
% Because the laws of motion of mechanical systems are often continuous functions in configuration space, they are of course representable as neural networks.
% (This \emph{a priori} truism has stood up to experimentation 
% [\cite{narendra_neural_1992,masri_identification_1993}].)
% But to what end?
% On one hand, neural networks can recover the laws of motion of a dynamic system from data for which no parametric form is known \cite{pillonetto_deep_2025}.
% On the other hand, neural networks can be supervised (``physics-informed'') using the numerical solution of another physical model that might prove more costly to evaluate \cite{mohajerin_multistep_2019,michalowska_neural_2024}.

This paper deals with the problem of filtering---state estimation from past and present output measurements.
Classic applications include tracking and predicting the ballistic motion of a point mass on the basis of noisy measurements such as radar,
 or estimating one's own location by fusing 
inertial sensors and satellite navigation.
Sundry applications include human physiology forecasting \citep{albers_interpretable_2023}, pandemic surveillance \citep{alsaggaf_nonlinear_2024}, and latent macroeconomic variables \citep{burmeister_kalman_1982}.
A restricted version of the filtering problem was solved in the '60s for linear systems with Gaussian process and measurement noise \citep{kalman_new_1961,luenberger_observing_1964}.
For nonlinear systems, the theoretical Bayesian filter is intractable.
The playing field for how to generalize the Kalman filter is wide and highly contested, and there are legion variational approximations such as the Extended, and Unscented Kalman filters which claim to handle nonlinearity more faithfully \citep{sarkka_bayesian_2023,jiang_new_2025}.
This paper's scope is limited to filters whose only state consists of a mean and covariance (thus excluding particle filters, ensemble Kalman filters, and PDE-based methods).

% \textbf{Premise 2.} \emph{Modeling dynamic systems enables prediction, filtering, and smoothing.}
% Born in the late 20th century, these now-classic problems deal with estimating the latent state of a partially observed Markov model.
% For example, tracking and predicting the ballistic motion of a point mass (such as a golf ball) on the basis of noisy measurements such as radar,
%  or estimating one's own location by fusing 
% inertial sensors and satellite navigation.
% Sundry applications include human physiology forecasting \cite{albers_interpretable_2023}, pandemic surveillance \cite{alsaggaf_nonlinear_2024}, and latent macroeconomic variables \cite{burmeister_kalman_1982}.
% \cite{kalman_new_1961} and \cite{luenberger_observing_1964} stated and solved the filtering problem for linear dynamic system models.
% For nonlinear models, the playing field for how to generalize the Kalman Filter is wide and hotly contested.
% It is currently occupied by the Extended, Unscented, and countless other Kalman Filters that claim to handle nonlinearity more faithfully \cite{sarkka_bayesian_2023,jiang_new_2025}.

Existing work on state estimation with neural network dynamic models has treated the network as a black box, embedding it inside a general-purpose nonlinear Kalman filter such as the Extended Kalman Filter \citep{oveissi_novel_2025} or the Unscented Kalman Filter \citep{anurag_rcukf_2025}.
(Also see \citet{bai_state_2023} for an ample bibliography.)
The basic challenge to these methods is propagating a Gaussian uncertainty through the dynamic model, as these methods are all perturbative in the small-variance limit.
% Therefore, all of these methods inherit the same weaknesses of nonlinear Kalman filters when it comes to strongly nonlinear systems.

% \textbf{Conclusion.} \emph{Therefore, neural networks enable prediction, filtering, and smoothing.}
% We are not the first to make this connection.
% Existing work has embedded the neural mapping inside an general-purpose nonlinear Kalman filter such as the Extended Kalman Filter \cite{oveissi_novel_2025} or the Unscented Kalman Filter \cite{anurag_rcukf_2025}.
% (Also see \cite{bai_state_2023} for an ample bibliography.)

% \subsection{Foreground}
% The syllogism above introduces the applicability of neural networks to state estimation.
% Now, in the same three steps, we introduce the innovation of this paper, which consists in applying the uncertainty propagation methodology of \neuralUQ to state estimation.

Our work, which constructs a bespoke analytic Kalman filter for neural network dynamics, begins by recognizing that a neural network is not just any smooth black-box function.
All of the nonlinearity of a neural network is concentrated in its activation function, which can be chosen to facilitate closed-form moment propagation \citep{neural-kalman-anonymous}.
% As \neuralUQ finds, neural networks are not only a trainable universal approximator but also a trainable uncertainty propagator.
With the right activation function, the first two moments of a Normal distribution can be propagated exactly through a single layer of a feedforward neural network, and by extension, approximately through a deep feedforward neural network.
Across a range of applications, this uncertainty propagation method dramatically outperforms linearized and unscented propagation.
The importance of moment accuracy is highlighted by \citet{deisenroth_analytic_2009}, which finds that analytic moment propagation improves Kalman filtering on models represented as Gaussian processes.
Our work turns to models represented as neural networks.

% \textbf{Premise 1.} \emph{Neural networks are effective for propagating uncertainty in dynamic systems.}
% As \neuralUQ finds, neural networks are not only a trainable universal approximator, but also, a trainable uncertainty propagator.
% With the right activation function, the first two moments of a Normal distribution can be propagated exactly through a single layer of a feedforward neural network, and by extension, approximately through a deep feedforward neural network.
% The importance of moment accuracy is highlighted by \cite{deisenroth_analytic_2009}, in which it is shown that analytic moment propagation is beneficial for Kalman filtering on models represented as Gaussian processes.
% Our work turns to models represented as neural networks.

Recursive Bayesian state estimation can be implemented using three fundamental operations on functions:
coupling (forming the joint distribution of two functions applied to the same input), conditioning, and uncertainty propagation.
This paper implements the coupling operation using the grammar of neural networks.
It imposes a variational approximation to simplify conditioning using the Gaussian formula.
It tests different options for the plug-in uncertainty propagation method, including linearized propagation, unscented propagation, and analytic propagation (introduced in \neuralUQ).
% The third is adopted from \neuralUQ.
% The various uncertainty propagation methods (unscented, linear, etc.) benchmarked in \neuralUQ correspond to different flavors of Kalman filtering (unscented, extended, etc.).

% \textbf{Premise 2.} \emph{Neural network uncertainty propagation enables prediction, filtering, and smoothing.}
% Bayesian prediction, filtering, and smoothing can be implemented using three fundamental operations:
% coupling (forming the joint distribution of two functions applied to the same input), conditioning, and uncertainty propagation.
% The various uncertainty propagation methods (unscented, linear, etc.) benchmarked in \neuralUQ correspond to different flavors of Kalman filtering (unscented, extended, etc.).
% While we leave the conditioning operation unchanged, we construct couplings of neural networks that are themselves neural networks (Lemma \ref{lem:augmentation}).

We find that analytic moment propagation through neural networks enables state estimation in problems hitherto believed to be intractable.
The literature exhibits a survivorship bias for nonlinear filtering and smoothing problems that are amenable to general-purpose nonlinear Kalman filters (as detailed in the Background section).
Usually, this amounts to selecting a sampling time that is small enough that the discretized dynamics are amenable to perturbative uncertainty propagation methods such as linearization.
For example, the Lorenz system is filtered using a discretization time of 0.001 \cite{nosrati_chaotic_2011}, 0.005 \cite{dubois_data-driven_2020}, or 0.01 \cite{oveissi_novel_2025}.
Our work extends the discretization time to 1, which is on the order of a full period of the chaotic attractor.

% \textbf{Conclusion.} \emph{Therefore, analytic moment propagation through neural networks enables prediction, filtering, and smoothing} of problems hitherto believed to be intractable.
% The literature exhibits a survivorship bias for nonlinear filtering and smoothing problems that are amenable to general-purpose nonlinear Kalman filters (as detailed in the Background section).
% Usually, this amounts to selecting a sampling time that is small enough that the discretized dynamics are amenable to perturbative uncertainty propagation methods such as linearization.
% For example, the Lorenz system is filtered using a discretization time of 0.001 \cite{nosrati_chaotic_2011}, 0.005 \cite{dubois_data-driven_2020}, or 0.01 \cite{oveissi_novel_2025}.
% Our work extends the discretization time to 1, which is on the order of a full period of the chaotic attractor.

We do not impeach the Extended, Unscented etc.~Kalman filters as unfit for the problems on which they have been applied.
% on existing problems, our method is only an incremental improvement.
Rather, by applying the analytic moment propagation method of \neuralUQ to a neural network representation of the dynamics, we unlock Kalman filtering on harder, more nonlinear problems.
On highly nonlinear problems, our method improves upon existing Kalman filters in both accuracy and calibration.
We show an improvement in the point prediction and estimate of the state (accuracy). Qualitatively, we observe that our estimator is able to stay synchronized with the chaotic Lorenz system.
We demonstrate that confidence regions predicted by our method's state covariances have better coverage than existing Kalman filters (calibration).
    That is, 95\% confidence regions for the state contain the true state roughly 95\% of the time.
    This enables risk-aware decision making for optimally trading between safety and performance.
% \begin{description}
%     \item[accuracy] We show an improvement in the point prediction and estimate of the state. Qualitatively, we observe that our estimator is able to stay synchronized with the chaotic Lorenz system.
%     \item[calibration] We demonstrate that confidence regions predicted by our method's state covariances have better coverage than existing Kalman filters.
%     That is, 95\% confidence regions for the state should contain the true state roughly 95\% of the time.
%     This enables risk-aware decision making for optimally trading between safety and performance.
% \end{description}

\section{Notation}
When \(\sigma:\mathbb R \to \mathbb R\) is a neural network activation function, \(\sigma (x)\) for \(x \in \mathbb{R}^n\) is applied elementwise.
% Uncertainty propagation notation follows that of \neuralUQ.
If \(X\) is a square-integrable random vector, the notation \(\normal X\) refers to a random variable having distribution \(\mathcal N(\expect X, \Cov X)\).
If \(f\) is a neural network, the notation \(\Normal f(X)\) is a Normal random variable in which the normality approximation is applied layer-by-layer as in the analytic propagation method of \neuralUQ.
The comma \(,\) is a higher-order function:
if \(x \mapsto f(x)\) and \(x \mapsto g(x)\) are functions, then \((f
, g)\) is the function \(x \mapsto (f(x), g(x))\).
% The direct sum of two square matrices \(A\) and \(B\) is denoted by
% \(A \oplus B\) and refers to the block-diagonal matrix \(A \oplus B =
%   \begin{pmatrix} A & 0 \\ 0 & B
% \end{pmatrix}\).

% In algorithmic pseudocode, a \emph{Function} is a pure function in
% the mathematical or functional
% programming sense.
% In particular, the variable names of arguments are  placeholders.
% \(f(x) = x^2\) defines the same function as \(f(y) = y^2\).

\section{Problem statement}
A dynamic system is described by
\begin{subequations}
  \label{eq:dynamic-system}
  \begin{align}
    x_0 &\sim \mathcal{N}(\mu_0, \Sigma_0) \\
    x_t &= F(x_{t-1}, u_t) + \eta_t, &\eta_t &\sim \mathcal{N}(0, Q)
    & \forall t &\in \cbr{1 \ldots T}\\
    y_t &= H(x_t, u_t) + \epsilon_t, &\epsilon_t &\sim \mathcal{N}(0,
    R) & \forall t &\in \cbr{1 \ldots T}
  \end{align}
\end{subequations}
where \(x_t \in \mathbb{R}^{n_x}\) is the state, \(u_t \in \mathbb
R^{n_u}\) is the input, and \(y_t \in \mathbb{R}^{n_y}\) is the output.
The random variables \(x_0\), \(\{\eta_t\}_{t=1}^T\), and
\(\{\epsilon_t\}_{t=1}^T\) are independent.

This model motivates three problems.
\begin{problem}[Prediction]\label{problem:prediction}
  Given \(\{u_s\}_{s=0}^{t-1}\) and \(\{y_s\}_{s=1}^{t-1}\), predict
  \(x_t\) and \(y_t\) as a joint distribution \((\hat X_{t \mid t-1},
  \hat Y_{t \mid t-1})\).
\end{problem}

\begin{problem}[Filtering]\label{problem:filtering}
  Given \(\{u_s\}_{s=0}^{t}\) and \(\{y_s\}_{s=1}^t\), estimate
  \(\hat X_{t \mid t}\).
\end{problem}

\begin{problem}[Smoothing]\label{problem:smoothing}
  Given \(\{u_s\}_{s=0}^T\) and \(\{y_s\}_{s=1}^T\), estimate \(\hat X_{t \mid T}\).
\end{problem}

In the case that \(F\) and \(H\) are linear functions,
the predictive and posterior distributions arising in these problems are Normal and can be computed analytically by recursion across time steps.
However, in our case, \(F\) and \(H\) are nonlinear functions.
Pursuant to the Assumed Density Filtering Ansatz \citep{deisenroth_analytic_2009}, 
we impose Normality assumptions on nonlinearly transformed Normal variables.
% This notation is used in order to make a common footing for the varieties of Assumed Density Filtering, which can be seen as different implementations of the ``\(\normal\)'' operator.

\paragraph{Abstract Kalman filtering}
The Kalman filter (\algorithmref{alg:kalman-filter}) solves the prediction
and filtering problems by a forward recursion.
% \algorithmref{alg:kalman-filter} has two phases.
In \texttt{Predict}, it uses the estimate of the current state \(X\) to compute the jointly Gaussian distribution \((X', Y')\) of the next state and next output.\footnote{
  Frequently (e.g.~in \citealt{jiang_new_2025}) this step is notated in math and implemented in code as three separate uncertainty propagations:
  first, the uncertainty of \(X\) is propagated through the dynamic model to get \(X'\);
  second, the uncertainty of \(X'\) is propagated through the observation model to get \(Y'\);
  third, the uncertainty of \(X'\) is propagated through the observation model again to compute the cross-covariance of \(X'\) and \(Y'\).
  Our method combines the latter two steps by using a neural network to express the coupling of \(X'\) and \(Y'\).
  This enables the analytic implementation of the neural network uncertainty propagation operator developed in \neuralUQ.
}
In \texttt{Update}, it incorporates the measurement \(y\) by conditioning the joint predictive distribution \((X', Y')\) on \(Y = y\).
We also benchmark against an optional refinement \texttt{Recal} developed in \cite{jiang_new_2025,jiang_mitigating_2026}, listed in Algorithm~\ref{alg:kalman-filter-recal}, which recalibrates the corrected state covariance using the posterior mean,
in cases in which \(H\) is a nonlinear function.

\paragraph{Abstract RTS smoothing}
The Rauch-Tung-Striebel smoother \algorithmref{alg:rts-smoother} solves the
smoothing problem by a backward recursion.
\algorithmref{alg:rts-smoother} is also implemented as \texttt{Predict} and \texttt{Update}.
In \texttt{Predict}, it uses the estimate of the current state \(X\) to compute the jointly Gaussian distribution \((X, X')\) of the current and next state.
In \texttt{Update}, it uses the estimate of the next state \(X'\) to update the estimate of the current state \(X\).

In Appendix \ref{sec:algorithms}, we list the pseudocode for these algorithms using high-level probabilistic notation similar to  \citet[Algorithms~5.3--4, 38.1--2]{hennig_probabilistic_2022}, in which the principal objects are probability distributions represented by capital letters.

\section{Neural networks}
We introduce a preferred formalism for neural networks.
The four-parameter \((A, b, C, d)\) layer function is more general than customary.
This generality allows for certain possibilities such as residual networks (\(C=I\), \(d=0\)), as well as computing certain joint distributions of crucial importance in Kalman filtering and RTS smoothing.

% In the context of the Kalman filter, this generality permits a convenient approximation of the joint distribution of the network's /input and its output, as the next section will show.

\begin{definition}
    \label{def:layer-function}
    A layer function is a function \(g:\mathbb R^n \times \mathbb R^{m \times n} \times \mathbb R^m \times \mathbb R^{m \times n} \times \mathbb R^m \to \mathbb R^m\) defined by \(g(x; A, b, C, d) = \sigma(A x + b) + C x + d\), where \(A \in \mathbb R^{m \times n}, b \in \mathbb R^m, C \in \mathbb R^{m \times n}, d \in \mathbb R^m\) are parameters.
\end{definition}

\begin{definition}
    \label{def:neural-network}
    A neural network with \(\ell \) layers is the function \(f: \mathbb R^{n_x} \to \mathbb R^{n_y}\) defined by
    \begin{align*}
        f(x) &= f^\ell(x) \\
        f^k(x) &= g(f^{k-1}(x); A^k, b^k, C^k, d^k) & \forall k &\in \cbr{1 \ldots \ell} \\
        f^0(x) &= x
    \end{align*}
\end{definition}

\subsection{The identity-augmentation operator}
\label{sec:identity-augmentation}
% Algorithms \ref{alg:kalman-filter} and \ref{alg:rts-smoother}
\algorithmref{alg:kalman-filter} requires the joint distribution of \((X_t, Y_t)\),
and \algorithmref{alg:rts-smoother} requires the joint distribution of \((X_t, X_{t+1})\).
Both of these distributions can be computed using the identity-augmentation operator \(F \mapsto F_\text{aug}=(\text{id}, F)\).
We construct a representation of \(F_\text{aug}\) that is itself a neural network.
This construction is not inherently profound, but becomes extremely useful (when paired with a general-purpose Normal approximation of a neural network's output) as a mathematical and programmatic way of computing the joint distribution of a neural network's input and its output.

\begin{lemma}
  \label{lem:augmentation}
  Neural networks defined by Def.~\ref{def:neural-network} are closed under the input coupling:
  if \(f_1\) and \(f_2\) are two neural networks with \(n\) inputs and \(\ell\) layers, then \((f_1, f_2)\) can also be parameterized by a neural network with \(n\) inputs and \(\ell\) layers.
\end{lemma}
The idea of the construction, detailed in \appendixref{app:proof-augmentation}, is to build block-identity \(C\) matrices.
The first layer repeats the input by \(x\mapsto(x,x)\), and the rest of the network is a direct sum of the layers of \(f_1\) and \(f_2\).
\begin{corollary}
  If \(f\) is a neural network with \(n\) inputs and \(\ell\) layers,
  then \((\operatorname{id}, f)\) can be represented by a neural network with \(n\)
  inputs and \(\ell\) layers.
\end{corollary}
\begin{proof}
  In order to appeal to Lemma \ref{lem:augmentation}, we just have to represent the identity map as a neural network with \(\ell\) layers.
  This can be done by setting \(A^k = 0\), \(b^k = 0\), \(C^k = I\) and \(d^k = 0\) for all \(k \in \cbr{1 \ldots \ell}\).
\end{proof}

\subsection{Methods for uncertainty propagation: with application to Kalman filtering}
Let \(X\) be a multivariate Normal random variable and \(F\) a neural network following Def.~\ref{def:neural-network}.
Uncertainty propagation refers to approximating\footnote{For a deep neural network, accurate uncertainty propagation is \(\sharp\)P-hard; see \S5, \neuralUQ.} the mean and covariance matrix of \(Y = F(X)\).
It appears as the ``\(\normal\)'' operator in Alg.~\ref{alg:kalman-filter}, lines 2--3.
% ; and Alg.~\ref{alg:rts-smoother}, line 2.
The introductory material draws on \citep[\S3]{neural-kalman-anonymous}.  
% , and we refer to that work for more details.

For any \(\sigma\), the exact moments of a single layer are given by three transcendental functions \(M_\sigma, K_\sigma, L_\sigma\) corresponding to bivariate Gaussian integrals:
\begin{lemma}[Lemma~1, \neuralUQ]
  For some activation function \(\sigma\), let \(g\) be the function defined by \(g_\sigma(x; A, b, C, d) = \sigma(A
  x+ b) + C x + d\).
  Let \(X \sim \mathcal N(\mu, \Sigma)\).
  Then
  \begin{align*}
    \del{\expect g_\sigma(X; A, b, C, d)}_i &= M_\sigma(\mu_i; \nu_{ii}) +
    (C\mu)_i + d_i
  \end{align*}
  and
  \begin{align*}
    \begin{split}
      \del{\Cov g_\sigma(X; A, b, C, d)}_{i, j} &=
        K_\sigma\del{
          \mu_i, \mu_j; \nu_{ii}, \nu_{jj}, \nu_{ij}
        } \\
        &\quad + L_\sigma\del{
          \mu_i; \nu_{ii}, \tau_{jj},\kappa_{ij}
        }
        \\
        &\quad + L_\sigma\del{
          \mu_j; \nu_{jj}, \tau_{ii},\kappa_{ji}
        } \\
        &\quad + \tau_{ij}.
    \end{split}
  \end{align*}
  where for all valid indices \((i, j)\),
  \begin{align*}
    \mu_i &= (A\mu + b)_i
    &
    \tau_{ij}
    &= (C\Sigma C^\intercal)_{i,j}
    \\
    \nu_{ij} &= (A\Sigma A^\intercal)_{i,j}
    &
    \kappa_{ij}
    &= (A \Sigma C^\intercal)_{i,j}
  \end{align*}
\end{lemma}

In particular, this work uses the sine activation function \citep{sitzmann_implicit_2020}, for which the moment maps are given by
\begin{align}
\MoveEqLeft
  M_{\sin}(\mu; \nu) = e^{-\nu/2} \sin(\mu)
  \tag{Lem.~6, \neuralUQ}
\end{align}
\begin{multline}
  K_{\sin}(\mu_1, \mu_2; \nu_{11}, \nu_{22}, \nu_{12}) 
  \\
  = \frac{1}{2} \sbr{e^{-\frac{\nu_{11} + \nu_{22}}{2} + \nu_{12}} - e^{-\frac{\nu_{11} + \nu_{22}}{2}}} \cos(\mu_1 - \mu_2) \\
  - \frac{1}{2} \sbr{e^{-\frac{\nu_{11} + \nu_{22}}{2} - \nu_{12}} - e^{-\frac{\nu_{11} + \nu_{22}}{2}}} \cos(\mu_1 + \mu_2)\\
  \tag{Lem.~7, \neuralUQ}
\end{multline}
\begin{gather}
  L_{\sin}(\mu_1, \mu_2; \nu_{11}, \nu_{22}, \nu_{12})
  = \nu_{12} e^{-\nu_{11}/2} \cos(\mu_1)
  \tag{Lem.~8, \neuralUQ}
\end{gather}

This paper introduces the \textbf{\textsc{analytic} Kalman filter}, in which ``\(\normal\)'' is implemented using the layer-by-layer moment matching method introduced in \neuralUQ, which provides analytical expressions for the cases where \(\sigma\) is a Normal CDF function or a sinusoid.
\begin{definition}
    Let \(f\) be a neural network with \(\ell\) layers.
    Given \(X \sim \mathcal N(\mu, \Sigma)\), the layer-wise Gaussian approximation of \(f(X)\), denoted \(Y_\mathrm{ana} = \Normal f(X)\), is the random variable defined by
    \begin{align*}
        Y_\mathrm{ana} &=  Y^\ell\\
        Y^k &= \normal g(Y^{k-1}; A^k, b^k, C^k, d^k) & \forall k &\in \cbr{1 \ldots \ell} \\
        Y^0 &= X
    \end{align*}
\end{definition}

Baseline methods are presented in \appendixref{app:baselines}.
Each of these methods is also tested with the \textbf{\textsc{recalibrate}} variation in which the \textsc{Update} step (Alg.~\ref{alg:kalman-filter}) is replaced by the recalibrate/back-out procedure described in \cite{jiang_mitigating_2026}, which claims that the \textsc{recalibrate} variation of nonlinear Kalman filtering confers a more conservative management of state uncertainty in the presence of strong nonlinearity.

% We also benchmark against the \textbf{\textsc{stationary}} state estimate, which predicts a constant \(\hat x\) for all \(t\):
% \begin{subequations}
% \begin{align}
%   \hat x_{\mathrm{stationary}} &= \mathcal N(\mu_{\mathrm{stationary}}, \Sigma_{\mathrm{stationary}})
%   \\
%   \mu_{\mathrm{stationary}} &= \frac{1}{T+1} \sum_{t = 0}^T x_t
%   \\
%   \Sigma_{\mathrm{stationary}} &= \frac{1}{T} \sum_{t = 0}^T (x_t - \mu_{\mathrm{stationary}})(x_t - \mu_{\mathrm{stationary}})^\intercal
% \end{align}
% \end{subequations}
% The meaning of \textbf{\textsc{stationary}} as a baseline is that it is the best state estimator that is agnostic to the the dynamics and ignores the measurements.
% Therefore, any Kalman filter with knowledge of the dynamics and access to the measurements should not perform worse than the \textsc{stationary} baseline.

\section{Performance criteria: beyond RMSE}
At each time step, Problems \ref{problem:prediction}, \ref{problem:filtering}, and \ref{problem:smoothing} call for a (predictive or posterior) distribution \(\hat X\) with mean \(\hat \mu\) and covariance \(\hat \Sigma\).
We also know the ground truth state \(x\) at each time step.
Thus the performance of the prediction, filtering, or smoothing \(\hat X\) should be understood via the joint distribution of \((x, \hat \mu, \hat \Sigma)\).

This section explores the choice of a criterion function \(J(x, \mu, \Sigma)\) to be averaged over the joint distribution of \((x, \hat \mu, \hat \Sigma)\).
If the underlying dynamic system is ergodic, then the phase average is also the time average;
the average-case performance of the Kalman filter is the same as the long-run performance.
% 
% How should we assess the quality of \(\hat X\) in light of the true state \(x\)?
% What should we choose as \(J(x, \mu, \Sigma)\)?
% 
\begin{definition}[RMSE]
The root-mean-square error is defined by \(J(x, \mu, \Sigma) = \norm{x - \mu}_2\).
\end{definition}
Most works we review treat the RMSE as the only figure of merit in Kalman filtering and ignore \(\Sigma\).\footnote{
For example, in \citet{jiang_mitigating_2026}, \(\Sigma\) does not attain a statistical meaning; its only responsibility is to generate the exploration region for the subsequent state update.
}
But it is frequently asserted that the distribution \(\mathcal N(\mu, \Sigma)\) should be interpreted as a Bayesian predictive or posterior distribution for \(x\).
Then it is not enough to have the right answer \(\mu\); the algorithm should also have the right uncertainty \(\Sigma\).
This agreement is called \emph{calibration}.
Calibration can be understood as a generalization of betting income on the outcome of a binary event \citep[Example~6.1.1]{cover_elements_2006}.
Linear filtering is calibrated, and state uncertainty can be used e.g.~to trade off between safety and performance in control applications \citep{kishida_risk-aware_2024}.
On the other hand, calibration failures have been blamed for the tendency of conversational large language models to assert confidently incorrect facts \citep{kalai_why_2025}.
Let us examine the implications of calibration.

First,
the triples \((x, \mu, \Sigma)\) should be appear to have come from a hierarchical model in which \(x \sim \mathcal N(\mu, \Sigma)\) where \((\mu, \Sigma)\) have the true marginal distribution of the latter two variables in \((x, \mu, \Sigma)\).%
\begin{subequations}%
\begin{align}%
  \mu, \Sigma &\sim \text{true distribution}
  \\
  x \mid \mu, \Sigma &\sim \mathcal N(\mu, \Sigma)
\end{align}%
\label{eq:hierarchical-model}%
\end{subequations}%
The goodness-of-fit of this model is summarized by the log likelihood \(\log p_{x \sim \mathcal N(\mu, \Sigma)}(x)\).
% The average value of this quantity can be interpreted as the KL divergence between the approximation \eqref{eq:hierarchical-model} and the true distribution of \((x, \mu, \Sigma)\).
This induces the cross entropy criterion, which is also used in \cite{deisenroth_analytic_2009}.
\begin{definition}[Cross entropy]
  The cross entropy criterion is defined by 
  \begin{align*}
    J(x, \mu, \Sigma)
    = \frac{1}{2} \log \det \Sigma + \frac{1}{2} (x - \mu)^\intercal \Sigma^{-1} (x - \mu).
  \end{align*}
\end{definition}
It can also be understood as the statistically optimal trade-off between confidence (small \(\Sigma\)) and caution (large \(\Sigma\)).
Holding \(\Sigma\) fixed, the cross entropy is minimized on average when \(\mu \approx x\).
Holding \(\mu\) fixed, the cross entropy is minimized on average when \(\Sigma \approx \expect{(x - \mu)(x - \mu)^\intercal}\).
This shows that a lower cross entropy combines point estimation error and uncertainty estimation error.

Second, model \eqref{eq:hierarchical-model} also implies
\((x -\mu)^\intercal \Sigma^{-1} (x-\mu)\) has the \(\chi^2\) distribution with \(n_x\) degrees of freedom.
Let \(F_{\chi^2(n_x)}\) be its cumulative distribution function.
Then for all \(\alpha \in [0, 1]\), it holds that
\begin{align*}
  \probability\sbr{
    (x -\mu)^\intercal \Sigma^{-1} (x-\mu) 
    \leq F_{\chi^2(n_x)}^{-1}(1-\alpha)
  }
  \geq 1 - \alpha.
\end{align*}
Thus \(\{x \mid (x -\mu)^\intercal \Sigma^{-1} (x-\mu) \leq F_{\chi^2(n_x)}^{-1}(1-\alpha)\}\) is a \((1-\alpha)\)-confidence set for \(x\) \citep[\S6.3.2]{wasserman_all_2004}, which motivates the coverage criterion.
\begin{definition}[Coverage]
  The \((1-\alpha)\)-coverage criterion is defined by
  \begin{align*}
    J(x, \mu, \Sigma)
    = \begin{cases}
      1, & (x -\mu)^\intercal \Sigma^{-1} (x-\mu) 
        \leq F_{\chi^2(n_x)}^{-1}(1-\alpha)\\
      0, & \text{otherwise}
    \end{cases}
  \end{align*}
\end{definition}
If \(\alpha = 0.05\), then the expected value of the \((1-\alpha)\)-coverage criterion is the long-run frequentist probability that a 95\% confidence set for \(x\) traps the true value of \(x\).
We simultaneously seek to minimize the volume of these confidence sets.
\begin{definition}[Coverage volume]
  The \((1-\alpha)\)-coverage volume criterion is defined by
  \begin{align*}
    J(x, \mu, \Sigma)
    = 
    \sbr{F_{\chi^2(n_x)}^{-1}(1-\alpha)}^{n_x/2}
    V_{n_x}
    \sqrt{\det \Sigma} ,
  \end{align*}
  where \(V_{n_x}\) is the volume of an \(n_x\)-ball.
\end{definition}

A filter can have a low RMSE and good coverage yet still be statistically suboptimal, as the
examples in \appendixref{app:motivating-examples} show.
The takeaway of this section, therefore, is that
cross entropy has the last word on the validity of a statistical inference.

\section{Example: stochastic Lorenz system estimation using a surrogate model}
\label{sec:stochastic-lorenz-system}
The Lorenz system is a reduced-order continuous-time model of atmospheric convection.
Here, we consider the problem of estimating all three states \((x^1, x^2, x^3)\) of the Lorenz system from a sequence of noisy and temporally sparse measurements of \(x^1\). 
The system of interest is the sampled stochastic Lorenz system with a deterministic initial condition.
\begin{subequations}
  \begin{align}
    x_0 &=  (-8, 4, 27) 
    \\
    x_{t} &= F(x_{t-1}, B_t) &\forall t \in \cbr{1 \ldots T}.
  \end{align}
\end{subequations}
\(B_t\) is an independent Wiener process for each \(t \in \cbr{1 \ldots T}\).
The transition function \(F\) is given by
\begin{align}
  F(x, B) &= \xi(\Delta t) \\
  \intertext{subject to}
  \xi(0) &= x \\
  \xi(t) &= \int_0^t f(\xi(s)) \dif s + \eta B(t) &\forall t \in [0, \Delta t]
\end{align}
where \(f\) is the Lorenz vector field. (See \appendixref{app:lorenz} for details.)

\subsection{Results}
The discussion rounds to three significant digits and refers to the full results in \appendixref{app:stochastic-lorenz-system-results}.

The \textsc{analytic} method achieves the lowest RMSE at 15.2 (prediction), 9.76 (filtering), and 9.58 (smoothing) (Table~\ref{tab:lorenz_results_rmse}).
The next best method, \textsc{unscented'95 (recal)}, has a prediction RMSE of \(16.2\), filtering RMSE of \(12.5\), and smoothing RMSE of \(14.9\).
On an RMSE basis, from prediction to filtering to smoothing, \textsc{linear} performs worse despite seeing more data (21.4, 71.1, 72.3).
This shows that the linearized dynamics don't contain enough information to assimilate new data.
Moreover, all of these RMSEs are greater than the population RMS variation of the Lorenz system; the EKF performs worse than doing nothing (ignoring the measurements and reporting the same state at every time step).
The \textsc{unscented'02} Kalman Filters performs yet an order of magnitude worse.
Compared to \textsc{unscented'95}, the highly localized hyperparameter tuning of \textsc{unscented'02} results in profound instability to the point of numerical indeterminacy (with covariance matrices having condition numbers on the order of \(10^{14}\)).

We realized only after experiments that \textsc{mean-field} uncertainty propagation is simply not eligible for recursive Bayesian inference at all, as its mean-field Ansatz nullifies the cross-covariance computations set forth in \S\ref{sec:identity-augmentation}.
The output of \textsc{mean-field} ends up being a stationary distribution with no updates.

The \textsc{analytic} methods' coverage meets or slightly exceeds the nominal 95\% at 98.5\% (prediction), 97.9\% (filtering), and 97.7\% (smoothing) (Table~\ref{tab:lorenz_results_coverage95}).
In contrast, \textsc{linear} performs much worse with 5.9\%, 6.1\%, and 1.2\% coverage, and \textsc{unscented'95} achieves 50.5\%, 53.5\%, and 35.8\% coverage respectively.
The disastrous coverage of these two methods is explained by their smaller confidence regions: for example, the \textsc{linear} smoother's 95\% confidence regions are on average 1000x smaller than those of the \textsc{analytic} smoother (Table~\ref{tab:lorenz_results_volume95}).

Finally, the cross entropy of \textsc{analytic} (10.1, 5.72, 5.67) is significantly lower than all of the others (Table~\ref{tab:lorenz_results_cross_entropy}).

Visual inspection of the trajectory and coverage plots (Figures~\ref{fig:Method.ANALYTIC-trajectory}--\ref{fig:Method.UNSCENTED1-coverage}) provides qualitative insights that complement these quantitative results.
The trajectory plots superimpose 100 time steps of the ground truth on top of the method's 90\% confidence region.
The \textsc{analytic} method (Figure~\ref{fig:Method.ANALYTIC-trajectory}) exhibits the tightest alignment with the ground truth.
In contrast, the \textsc{linear} method (Figure~\ref{fig:Method.LINEAR-trajectory}) is confidently incorrect.
The \textsc{unscented} variants (Figures~\ref{fig:Method.UNSCENTED0-trajectory}--\ref{fig:Method.UNSCENTED1-trajectory}) demonstrate intermediate performance.

The \textsc{analytic} method achieves close to nominal coverage across all confidence levels (Figure~\ref{fig:Method.ANALYTIC-coverage}), whereas the \textsc{linear} method exhibits drastic undercoverage at all confidence levels (Figure~\ref{fig:Method.LINEAR-coverage}).

In summary, the \textsc{analytic} method results in truer point predictions and point estimates as well as truer confidence regions than all of the other methods, by up to a millionfold.

\section{Example: LTI dynamics with nonlinear output}
This section deals with the discrete-time Wiener system:
\begin{align}
  x_{t+1} &= A x_t + B u_t + \eta_t &\eta_t &\sim \mathcal{N}(0, Q)\\
  y_t &= H(x_t, u_t) + \epsilon_t &\epsilon_t &\sim \mathcal{N}(0, R)
\end{align}
where \(x_{t} \in \mathbb{R}^{n_x}\), \(u_t \in \mathbb{R}^{n_u}\), \(y_t \in \mathbb{R}^{n_y}\), \(A \in \mathbb{R}^{n_x \times n_x}\), \(B \in \mathbb{R}^{n_x \times n_u}\), \(Q \in \mathbb{R}^{n_x \times n_x}\), \(R \in \mathbb{R}^{n_y \times n_y}\), \(H: \mathbb{R}^{n_x + n_u} \to \mathbb{R}^{n_y}\) is a known strongly nonlinear function represented as a neural network, and \(\eta_t\) and \(\epsilon_t\) are independent noise terms.
% The parameters of \(H\) are initialized to excite the saturation regime of the probit activation function \(\Phi\).

\subsection{Prediction, filtering, smoothing of a stable system}
\label{sec:lti-estimation}
We consider the case where \(A\) is stable, \(u\) is a sinusoidal signal, and the objective is to estimate \(\{x_t\}\).
The details on this system and its simulation are given in Appendix~\ref{app:lti-estimation}.
Full results are listed in \appendixref{app:lti-estimation-results}.

In brief, the \textsc{analytic} method scores the lowest RMSE across prediction (1.38), filtering (1.31), and smoothing (0.994) (Table~\ref{tab:results_rmse}).
In second place is \textsc{unscented'95}, which scores 6.56, 6.54, and 6.61, respectively.
Without recalibration, the worst RMSE is achieved by \textsc{linear} (31.5, 31.5, 32.0), with \textsc{unscented'02} slightly better (24.4, 24.6, 25.6).

Calibration wise, without recalibration \textsc{linear} scores the worst coverage, with 28.4\% (prediction), 25.8\% (filtering), and 15.4\% (smoothing) (Table~\ref{tab:results_coverage95}); despite drawing confidence regions comparable in size to those of \textsc{analytic}, it captures the true state far less often.
The \textsc{unscented'02} method is similarly overconfident, covering only 57.6\%, 55.9\%, and 45.5\%.
In light of the moderate point estimation error, the poor coverage of these methods is striking (Figures~\ref{fig:trajectory-linear-no}--\ref{fig:trajectory-linear-yes}).

The \textsc{analytic} method results in well-calibrated 95\% confidence regions in all three problems, with 93.8\% actual coverage in smoothing (Table~\ref{tab:results_coverage95}).
The \textsc{unscented'95} smoother achieves 73.7\% true coverage but with smaller confidence volumes than \textsc{analytic} (Table~\ref{tab:results_volume95}).

Balancing these competing priorities with the cross entropy criterion, we see that \textsc{analytic} dominates among non-recalibrated methods, with cross entropies of $-1.73$, $-0.47$, and $0.11$ (Table~\ref{tab:results_cross_entropy}).
The \textsc{linear}, \textsc{unscented'95}, and \textsc{unscented'02} methods produce catastrophic cross entropies on the order of $10^2$--$10^5$, reflecting severe overconfidence.

\paragraph{Effect of covariance recalibration.}
We also evaluate each method under the covariance-recalibrated framework of \citet{jiang_new_2025}, which is claimed to improve Kalman filtering with a nonlinear \(H\) by re-evaluating the covariance around the updated state.
Figures~\ref{fig:recal-effect-rmse}--\ref{fig:recal-effect-cross-entropy} summarize how each metric changes between the non-recalibrated and recalibrated variants of each method.

For \textsc{analytic}, recalibration has a small effect on RMSE (e.g., $0.994 \to 0.900$ in smoothing), while coverage improves slightly from 94.3\% to 95.8\% in prediction and from 93.8\% to 96.3\% in smoothing (Figures~\ref{fig:recal-effect-rmse} and~\ref{fig:recal-effect-coverage95}).
This is consistent with the theoretical prediction of \citet{jiang_new_2025} that recalibration becomes unnecessary when the uncertainty propagation is already accurate, since the re-approximated quantities then agree with the originals.
In cross entropy, recalibration improves \textsc{analytic}'s smoothing score from $0.11$ to $-4.64$, yielding a well-separated first place (Fig.~\ref{fig:recal-effect-cross-entropy}); the effect on prediction and filtering is smaller than \textsc{analytic}'s first-place margin over the other methods.
We do not have an explanation for this observation.

By contrast, recalibration produces huge improvements for the less accurate propagation methods.
For \textsc{unscented'95}, it closes most of the gap to \textsc{analytic}: prediction RMSE improves from $6.56$ to $1.68$, coverage rises from 82.3\% to 93.3\%, and cross entropy drops from $918$ to $-1.69$.
For \textsc{unscented'02}, RMSE drops by nearly an order of magnitude (e.g., $24.4 \to 2.98$ in prediction), coverage jumps from 57.6\% to 91.0\%, and cross entropy falls from $1.1 \times 10^4$ to $0.097$.
For \textsc{linear}, RMSE is cut by roughly a third (e.g., $31.5 \to 9.83$ in prediction), coverage rises from 28.4\% to 68.2\%, and cross entropy drops by two orders of magnitude (e.g., $2.4 \times 10^4 \to 103$).
These results confirm the central claim of \citet{jiang_new_2025}: the conventional Kalman update systematically underestimates posterior covariance when the uncertainty propagation is approximate, and re-evaluating the covariance around the updated state substantially mitigates this overconfidence.

The takeaway of this experiment is twofold.
First, even when conventional methods can achieve a low RMSE, achieving statistical calibration takes more work, and \textsc{analytic} gives the best tradeoff between sensitivity and specificity in its uncertainty estimates.
Second, the benefits of covariance recalibration are most pronounced for methods with less accurate uncertainty propagation; when combined with the accurate moment propagation of \textsc{analytic} the benefits are less consistent, but can still yield best-in-class results.

\subsection{Quadratic regulation of an unstable system}
\label{sec:lti-regulation}
We consider the case where \(A\) is marginally stable, and the objective is to 
minimize a quadratic cost function of \(\{x_t\}\) and \(\{u_t\}\).
The control law is the linear quadratic regulator \(u=-K \hat x\), where \(\hat x\) is the Kalman filter estimate.
The details on this system and its simulation are given in Appendix~\ref{app:lti-regulation}.
Full results are listed in \appendixref{app:lti-regulation-results}.

Without recalibration, the total quadratic cost incurred by the \textsc{analytic}-powered controller
is only 6\% higher than that of the optimal linear quadratic regulator using linear state feedback (Table~\ref{tab:results_regulation}).
All of the other state estimation methods (\textsc{linear}, \textsc{unscented'95}, \textsc{unscented'02}) result in an effectively unstable closed-loop system, with costs of 8\,500--31\,000 times the optimal LQR.
This degradation is explained by inferior state estimation: these filters exhibit 400--680 times the RMSE of \textsc{analytic} (Table~\ref{tab:results_regulation_rmse}), and their 95\% confidence regions cover the true state less than 2\% of the time, compared to 90.4\% for \textsc{analytic} (Table~\ref{tab:results_regulation_coverage95}).

\paragraph{Effect of covariance recalibration.}
For \textsc{analytic}, recalibration has negligible effect on the regulation cost (1.06$\times$ LQR with or without), consistent with the estimation results above.

For the other methods, recalibration has a striking but mixed effect on closed-loop performance.
\textsc{unscented'95 (recal)} sees the most dramatic improvement: total cost drops from 10\,200$\times$ to 10.2$\times$ LQR, RMSE falls from 16.6 to 0.49, and coverage rises from 1.7\% to 85.2\%.
\textsc{unscented'02 (recal)} also improves substantially, with cost dropping from 8\,500$\times$ to 342$\times$ LQR, RMSE from 14.9 to 2.96, and coverage from 1.7\% to 94.6\%.
By contrast, \textsc{linear (recal)} does not benefit: its total cost actually increases from 31\,000$\times$ to 220\,000$\times$ LQR, suggesting that the linearization errors in \textsc{linear} are too severe for recalibration alone to rescue the closed-loop system.

% The regulation experiment reinforces the conclusion from the estimation experiment.
% Covariance recalibration can rescue methods whose uncertainty propagation is moderately inaccurate (\textsc{unscented'95}, \textsc{unscented'02}), but cannot compensate for fundamentally poor propagation (\textsc{linear}).
% Meanwhile, \textsc{analytic} achieves near-optimal regulation without recalibration, underscoring that accurate uncertainty propagation is the most important factor for both estimation and control.

\section{Conclusion}
The novelty of this work is that it applies analytic neural network uncertainty propagation to Kalman filtering and RTS smoothing by explicitly constructing input couplings as neural networks.
We view the different Kalman filter variants as induced by different versions of uncertainty propagation, which become interchangable implementations of the same abstract Kalman filter.
Also, this work raises the issue of calibration in nonlinear Kalman filtering and RTS smoothing, which is needed to ensure that the posterior distributions' Gaussian approximations are valid.

The significance of this work is that state estimation powered by neural network uncertainty propagation delivers more accurate state estimates, with better-calibrated uncertainty, than existing methods.
Our numerical examples exemplify this dramatic improvement by not only achieving state-of-the-art performance on a challenging Lorenz system identification problem, but also by delivering more accurate state estimates than existing methods on a linear system estimation problem.
In applications, this means more nonlinear systems, longer sampling periods between measurements, greater process and observation noise, weaker measurements, more optimal risk-aware decisions, and more optimal state-feedback controllers.

\acks{This work was supported by the National Science Foundation CAREER Program (Grant No. 2046292).}

\bibliography{neural-uq,nn-filtering}

\clearpage
\appendix
\section*{Supplementary material}
\tableofcontents

\clearpage
\listoftables

\clearpage
\listoffigures

\clearpage
\section{Pseudocode for Kalman filter and RTS smoother}
\label{sec:algorithms}
\begin{algorithm2e}[H]%
% \LinesNumbered
  \caption{%
    \label{alg:kalman-filter}%
    General Kalman algorithm for recursive \textbf{prediction} (problem
    \ref{problem:prediction}) and \textbf{filtering} (problem
  \ref{problem:filtering})}
  \DontPrintSemicolon
  \SetKwProg{Fn}{Function}{}{}
  \SetKwFunction{FPredict}{Predict}
\SetKwFunction{FUpdate}{Update}
\SetKwFunction{FFilter}{Filter}
\SetKwInOut{Input}{Input}
\Input{State transition function \(F\) and observation model \(H\)}
\Input{State covariance \(Q\) and observation covariance \(R\)}
\Input{Recalibrate flag \(\texttt{recal} \in \{\texttt{true}, \texttt{false}\}\)}
\Fn{\FPredict{$X, u$}}{
  \tcc{Propagate \(X\) through state transition}
  $X' \gets \normal F(X, u) + \mathcal{N}(0, Q)$
  \;
  \tcc{Propagate \(X'\) through observation model}
  $((X', u), Y') \gets \normal (\text{id}, H)(X', u) + \mathcal{N}(0, 0_{n_x} \oplus R)$ 
  \;
  \tcc{Get joint distribution of next state and next output}
  $(X', Y') \gets\ $ projection of \(((X', u), Y')\)\;    
  \Return$(X', Y')$ 
}
\Fn{\FUpdate{$(X, Y), y$}}{
  \tcc{Apply Bayes' rule}
  $(X', y) \gets$ conditional distribution of $(X,Y)$ given $Y=y$
  \;
  \Return$X'$
}
\Fn{\FFilter{$u_1, \ldots, u_t, y_1, \ldots, y_t$}}{
  $\hat X_{0 \mid 0} \gets \mathcal{N}(\mu_0, \Sigma_0)$
  \;
  \For{$k \in [1 \ldots t]$}{
    \tcc{Solution to problem \ref{problem:prediction}}
    $(\hat X_{k \mid k-1}, \hat Y_{k \mid k-1}) \gets \texttt{Predict}(\hat X_{k-1 \mid k-1}, u_k)$
    \;
    \tcc{Solution to problem \ref{problem:filtering}}
    $\hat X_{k \mid k} \gets \texttt{Update}((\hat X_{k \mid k-1}, \hat Y_{k \mid k-1}), y_k)$
    \;
    \tcc{Recalibration}
    \lIf{\texttt{recal}}{$\hat X_{k \mid k} \gets \texttt{Recal}((\hat X_{k \mid k-1}, \hat Y_{k \mid k-1}), \hat X_{k \mid k}, y_k, u_k)$}
  }
  \Return$\{ \hat X_{k \mid k} \}_{k=1}^t$
}
\end{algorithm2e}

\begin{algorithm2e}[H]%
  \caption{%
    \label{alg:kalman-filter-recal}%
    Recalibrated update step generalized from \citet{jiang_new_2025}}
  \DontPrintSemicolon
  \SetKwProg{Fn}{Function}{}{}
  \SetKwFunction{FRecal}{Recal}
  \SetKwFunction{FPredict}{Predict}
  \SetKwInOut{Input}{Input}
\Fn{\FRecal{$(X, Y), X', y, u$}}{
  \tcc{Kalman prediction (Algorithm \ref{alg:kalman-filter}) using corrected mean}
  $(X^+, Y^+) \gets \texttt{Predict}(\mathcal N(\expect X', \Cov X), u)$
  \;
  \tcc{Kalman gain}
  $K \gets \Cov(Y, X)\Cov(Y, Y)^{-1}$
  \;
  \tcc{Recalibrated covariance}
  $P^+ \gets \begin{pmatrix} I \\ -K \end{pmatrix}^\intercal \begin{pmatrix} \Cov(X, X) & \Cov(X^+, Y^+) \\ \Cov(Y^+, X^+) & \Cov(Y^+, Y^+) \end{pmatrix} \begin{pmatrix} I \\ -K \end{pmatrix}$
  \;
  \tcc{Back out if recalibration increased uncertainty}
  \lIf{$\trace P^+ > \trace \Cov X$}{
    \Return$\mathcal N(\expect X, \Cov X)$
  }
  \Return$\mathcal N(\expect X', P^+)$
}
\end{algorithm2e}
\clearpage

\begin{algorithm2e}[H]
  \caption{
    \label{alg:rts-smoother}
    General RTS algorithm for recursive \textbf{smoothing} (problem
    \ref{problem:smoothing})}
  \DontPrintSemicolon
  \SetKwProg{Fn}{Function}{}{}
  \SetKwFunction{FPredict}{Predict}
  \SetKwFunction{FUpdate}{Update}
  \SetKwFunction{FSmooth}{Smooth}
  \SetKwInOut{Input}{Input}
  \Input{State transition function \(F\)}
  \Input{State covariance \(Q\)}
  
  \Fn{\FPredict{$X, u$}}{
    \tcc{Propagate \(X\) through state transition}
    $((X, u), X') \gets \normal (\text{id}, F)(X, u) + \mathcal{N}(0, 0 \oplus Q)$
    \;
    \tcc{Get joint distribution of current and next state}
    \Return$(X, X')$
  }
  
  \Fn{\FUpdate{$(X, X'), X''$}}{
    \tcc{Apply Bayes' rule}
    $(X, X'') \gets$ conditional distribution of $(X,X')$ given $X' = X''$
    \;
    \Return$X$
  }
  
  \Fn{\FSmooth{$u_1, \ldots, u_T, y_1, \ldots, y_T$}}{
    $\{\hat X_{k|k}\}_{k=1}^T \gets \texttt{Filter}(u_1, \ldots, u_T, y_1, \ldots, y_T)$
    \;
    \For{$k \in [T-1, \ldots, 0]$}{
      \tcc{Compute joint distribution of current and next state}
      $(\hat X_{k|k}, \hat X_{k+1|k}) \gets \texttt{Predict}(\hat X_{k \mid k}, u_{k+1})$
      \;
      \tcc{Solution to problem \ref{problem:smoothing}}
      $\hat X_{k|T} \gets \texttt{Update}((\hat X_{k|k}, \hat X_{k+1|k}), \hat X_{k+1|T})$
      \;
    }
    \Return$\{\hat X_{k|T}\}_{k=1}^T$
  }
\end{algorithm2e}

\clearpage
\section{Motivating examples for confidence calibration}
\label{app:motivating-examples}

Let us consider three toy examples of joint distributions \((x, \mu, \Sigma)\) having a suboptimal \(\Sigma\). 
We will write \(J = J(\epsilon, \Sigma)\) with \(\epsilon = x - \mu\).
Every example has the same RMSE.
We vary only the law of \(\Sigma\) and its coupling with \(\epsilon\) in order to show that 
they make all the difference.
\begin{example}
  Let \(\epsilon \sim \mathcal N(0, I)\) and \(\Sigma = \frac{1}{2}I\).
  Then the 95\% coverage criterion has expected value strictly less than \(0.95\).
\end{example}
\begin{example}
  Let \(\epsilon \sim \mathcal N(0, I)\) and \(\Sigma = 2I\).
  Then the 95\% coverage criterion has expected value strictly greater than \(0.95\).
\end{example}
\begin{example}
  Let \(\epsilon \sim \mathcal N(0, I)\) and, independently, \(\Sigma\) obeys
  \begin{align*}
    \Sigma &= \begin{cases}
      0 I & \text{with probability } 0.05\\
      \infty I & \text{with probability } 0.95
    \end{cases}
  \end{align*}
  Then the 95\% coverage criterion has expectation 95\%, which is perfect.
  But no other confidence level is correctly covered, and the coverage volume is too large.
  The cross entropy is \(+\infty\).
\end{example}
This last example shows that coverage is necessary but not sufficient for a good predictive (posterior) distribution.
We desire that \(x\) is trapped by \(\Sigma\)'s confidence intervals not only in an average sense, but also conditional on \(x\).
\begin{example}
  Consider the following mixture model.
  We flip an unfair coin that comes up heads with probability 0.05.
  If heads, then \(\epsilon \sim \mathcal N(0, 1000I)\), \(\Sigma_1 = 0.001I\), and \(\Sigma_2 = 1000I\).
  If tails, then \(\epsilon \sim \mathcal N(0, I)\), \(\Sigma_1 = 1000I\), and \(\Sigma_2 = I\).

  Both \((\epsilon, \Sigma_1)\) and \((\epsilon, \Sigma_2)\) approximately achieve 95\% coverage.
  But the clearly superior \(\Sigma_2\) minimizes the cross entropy criterion.
\end{example}

\clearpage
\section{Proof of Lemma \ref{lem:augmentation}}
\label{app:proof-augmentation}
For \(j \in \{1, 2\}\), let \(f_j\) be defined by
\begin{align}
  f_j(x) &= f_j^\ell(x), \\
  f_j^k(x) &= g(f_j^{k-1}(x); A_j^k, b_j^k, C_j^k, d_j^k), & k \in \cbr{1 \ldots \ell},
  \\
  f_j^0(x) &= x
\end{align}
Now define \(f_\text{aug} =(f_1, f_2)\) by
\begin{align}
  f_\text{aug}(x) &= f_\text{aug}^\ell(x), \\
  f_\text{aug}^k(x) &= g(f_\text{aug}^{k-1}(x); A_\text{aug}^k, b_\text{aug}^k, C_\text{aug}^k, d_\text{aug}^k), & k \in \cbr{1 \ldots \ell},
  \\
  f_\text{aug}^0(x) &= x
\end{align}
where
\begin{align}
  A_\text{aug}^1 &= \begin{pmatrix}
    A_1 \\ A_2 
  \end{pmatrix}
  &
  b_\text{aug}^1 &= \begin{pmatrix}
    b_1 \\ b_2
  \end{pmatrix}
  \\
  C_\text{aug}^1 &= \begin{pmatrix}
    C_1 \\ C_2
  \end{pmatrix}
  &
  d_\text{aug}^1 &= \begin{pmatrix}
    d_1 \\ d_2
  \end{pmatrix}
\end{align}
and for \(k \in \cbr{2 \ldots \ell}\),
\begin{align}
  A_\text{aug}^k &= \begin{pmatrix}
    A_1^k & 0 \\
    0 & A_2^k
  \end{pmatrix}
  &
  b_\text{aug}^k &= \begin{pmatrix}
    b_1^k \\ b_2^k
  \end{pmatrix}
  \\
  C_\text{aug}^k &= \begin{pmatrix}
    C_1^k & 0 \\
    0 & C_2^k
  \end{pmatrix}
  &
  d_\text{aug}^k &= \begin{pmatrix}
    d_1^k \\ d_2^k
  \end{pmatrix}
\end{align}
% \clearpage
\section{Filtering and Smoothing baseline methods}
\label{app:baselines}
The implementation of ``\(\normal\)'' is varied to generate baseline Kalman filters.
Inspired by \citet{huber_bayesian_2020, wagner_kalman_2022,akgul_deterministic_2025}, we define a \textbf{\textsc{mean-field} analytic Kalman filter} by assuming that neurons in the same hidden layer are independent:
\begin{align*}
  Y_\mathrm{mfa} &=  Y^\ell\\
  Y^k &= \mathcal N(\mu^k, \Sigma^k) \quad \forall k \in \cbr{1 \ldots \ell} \\
  \mu^k &= \expect g(Y^{k-1}; A^k, b^k, C^k, d^k)
  \\
  \Sigma^k_{ij} &= \begin{cases}
    \sbr{\Cov g(Y^{k-1}; A^k, b^k, C^k, d^k)}_{ij}, & i=j
    \\
    0, &\text{else}
  \end{cases}
  \\
  Y^0 &= X
\end{align*}
The \textbf{\textsc{extended} Kalman filter} in works such as \citet{jiang_mitigating_2026,oveissi_novel_2025} uses a linearization of the neural network
\citep{titensky_uncertainty_2018, nagel_kalman-bucy-informed_2022,petersen_uncertainty_2024, jungmann_analytical_2025}.
The \textbf{\textsc{unscented'95} Kalman filter} introduced in \citet{julier_new_1995,julier_new_1997,julier_new_2000} is a one-parameter family of transformations that approximate the distribution of \(X\) by \(2n +1\) point masses.
The \textbf{\textsc{unscented'02} Kalman filter} introduced in \citet{julier_scaled_2002,wan_unscented_2000} adds two additional hyperparameters.
It is more commonly used today \citep{jiang_new_2025,anurag_rcukf_2025} and is the default in off-the-shelf software \citep{ljung_unscentedkalmanfilter_2025}.

\clearpage
\section{Supplement to \S\ref{sec:stochastic-lorenz-system}}
\label{app:lorenz}
The Lorenz vector field is given by
\begin{subequations}
\label{eq:stochastic-lorenz-system}
\begin{align}
  f\begin{pmatrix}
    x^1\\
    x^2\\
    x^3
  \end{pmatrix}
  &= \begin{pmatrix}
    \sigma (x^2 - x^1)\\
    x^1(\rho - x^3) - x^2\\
    x^1 x^2 - \beta x^3
  \end{pmatrix}
\end{align}
\end{subequations}
where \(\sigma = 10\), \(\rho = 28\), and \(\beta = 8/3\).
The process noise standard deviation is \(\eta = 0.001\).
The measurement noise standard deviation is \(\epsilon = 0.1\).
The sampling time is \(\Delta t = 1.0\).
% Trajectories visualized in Figure~\ref{fig:discretization} shows both the long duration of our \(\Delta 1\) in relation to previous works with \(\Delta t \ll 1\) and the sensitivity of the Lorenz dynamics to small perturbations.

In order to learn the generative model \ref{eq:dynamic-system} from a
training set \(\{x_t\}_{t =0}^{T_\text{train}}\),
we learn the neural network \(F\) and process
covariance \(Q\) simultaneously by maximizing the profile likelihood:
\begin{alignat}{2}
\label{eq:profile-likelihood}
    &\min_{F, Q}
    \quad
    && \det Q
    \\
    &\text{subject to}
    \quad
    && Q = \frac{1}{T_\text{train} - 1} \sum_{t=2}^{T_\text{train}}
            \epsilon_t \epsilon_t^\intercal
    \notag\\
    & && \epsilon_t = x_t - F(x_{t-1}, u_t)
    \notag
\end{alignat}

\subsection*{Implementation details}
We generated data using the stochastic integrator \citep{foster_high_2023} implemented by Diffrax \citep{kidger_neural_2022} within the JAX programming system \citep{deepmind_deepmind_2020,kidger_equinox_2021,bradbury_jax_2018}.
The training and validation datasets both had length \(T_\text{train} =200,000\).
The neural network had five hidden layers with 64 units each.
The activation function was sine, and there were residual connections between hidden layers.

We optimized \eqref{eq:profile-likelihood} using 20,000 epochs of AdamW \citep{loshchilov_decoupled_2019} implemented in Optax \citep{deepmind_deepmind_2020}.
% with the AdamW optimizer .
We used a minibatch size of 10,000 and decayed the learning rate from \(10^{-4}\) to \(10^{-5}\) over 20,000 epochs.
% The training 
% This took about ten hours.
% 
% Generating the numerical results for \S\ref{sec:stochastic-lorenz-system} took about six hours.
% 
This took about four hours on a Nvidia T1200 GPU, an Intel® Core™ i7-11850H CPU, and 32GB of RAM.

We report mean \(\pm\) standard error quantities based on twenty independent realizations of the test data.

\clearpage
\section{Supplement to \S\ref{sec:lti-estimation}}
\label{app:lti-estimation}
We generated the dynamic system with the following parameters:
\begin{itemize}
    \item State dimension: \( n_x = 5 \)
    \item Output dimension: \( n_y = 3 \)
    \item Input dimension: \( n_u = 1 \)
    \item \(A\) and \(B\) are in controllable canonical form with eigenvalues \( (0.9, 0.7, 0.5, 0.3, 0.1) \).
    \item \(H\) is a neural network with one hidden layer of 50 units and a \(\Phi\) activation, followed by a linear layer.
    The weights and biases are randomly initialized from normal distributions with default scaling.
    The second linear layer has zero weights and zero biases.
    \(H\) is trained to interpolate \(M=100\) random input-output pairs using the CoCoB optimizer for 10,000 steps.
    \item Noise covariances: 
        \begin{itemize}
            \item Process noise: \( Q = 10^{-3} I_{n_x} \)
            \item Measurement noise: \( R = 10^{-3} I_{n_y} \)
        \end{itemize}
\end{itemize}

The input is \(u(t) = \sin(0.2 t)\).
The initial state is \(0\) and the initial state estimate is \(\mathcal{N}(0, 0)\).

The horizon is \(T = 10000\) time steps.
We report mean \(\pm\) standard error quantities based on twenty independent realizations of the noises \(\eta_t\) and \(\epsilon_t\).

\clearpage
\section{Supplement to \S\ref{sec:lti-regulation}}
\label{app:lti-regulation}
We generated the dynamic system with the following parameters:
\begin{itemize}
    \item State dimension: \( n_x = 4 \)
    \item Output dimension: \( n_y = 8 \)
    \item Input dimension: \( n_u = 1 \)
    \item \(A\) and \(B\) are in controllable canonical form with eigenvalues \( (1, -1, 0.1, -0.1) \).
    \item \(H\) is a neural network with one hidden layer of 50 units and a \(\Phi\) activation, follow by another affine layer and a \(\Phi\) activation.
    The weights and biases are randomly initialized from normal distributions.
    In the first hidden layer, the weights are scaled by ten times the inverse square root of the number of inputs, and the biases are scaled by the inverse square root of the number of inputs.
    In the second hidden layer, the weights are scaled by the inverse square root of the number of inputs and the biases are set to zero.
    \item Noise covariances: 
        \begin{itemize}
            \item Process noise: \( Q = 10^{-2} I_{n_x} \)
            \item Measurement noise: \( R = 10^{-4} I_{n_y} \)
        \end{itemize}
\end{itemize}

The control objective is to minimize \(\sum_t \left\|x_t\right\|^2 + \sum_t \left\|u_t\right\|^2\).
The initial state is \(0\) and the initial state estimate is \(\mathcal{N}(0, 0)\).

The horizon is \(T = 10000\) time steps.
We report mean \(\pm\) standard error quantities based on twenty independent realizations of the noises \(\eta_t\) and \(\epsilon_t\).

\clearpage
\section{Full results for \S\ref{sec:stochastic-lorenz-system}}
\label{app:stochastic-lorenz-system-results}
\begin{table}[htbp!]
\centering
\caption{RMSE in the Lorenz system state estimation problem}
\label{tab:lorenz_results_rmse}
\begin{tabular}{lll}
\toprule
Method & Task & Value $\pm$ Monte Carlo standard error \\
\midrule
\textsc{{\textsc{analytic}}} & Prediction $(t|t-1)$ & \num[print-zero-exponent = true,print-implicit-plus=true,print-exponent-implicit-plus=true]{1.525e+01} \ensuremath{\pm} \num[print-zero-exponent = true,print-exponent-implicit-plus=true]{1.1e-02} \\
 & Filtering $(t|t)$ & \num[print-zero-exponent = true,print-implicit-plus=true,print-exponent-implicit-plus=true]{9.762e+00} \ensuremath{\pm} \num[print-zero-exponent = true,print-exponent-implicit-plus=true]{9.9e-03} \\
 & Smoothing $(t|T)$ & \num[print-zero-exponent = true,print-implicit-plus=true,print-exponent-implicit-plus=true]{9.579e+00} \ensuremath{\pm} \num[print-zero-exponent = true,print-exponent-implicit-plus=true]{1.0e-02} \\
 &  &  \\
\textsc{{\textsc{mean-field}}} & Prediction $(t|t-1)$ & \num[print-zero-exponent = true,print-implicit-plus=true,print-exponent-implicit-plus=true]{2.007e+01} \ensuremath{\pm} \num[print-zero-exponent = true,print-exponent-implicit-plus=true]{2.2e-02} \\
 & Filtering $(t|t)$ & \num[print-zero-exponent = true,print-implicit-plus=true,print-exponent-implicit-plus=true]{2.007e+01} \ensuremath{\pm} \num[print-zero-exponent = true,print-exponent-implicit-plus=true]{2.2e-02} \\
 & Smoothing $(t|T)$ & \num[print-zero-exponent = true,print-implicit-plus=true,print-exponent-implicit-plus=true]{2.007e+01} \ensuremath{\pm} \num[print-zero-exponent = true,print-exponent-implicit-plus=true]{2.2e-02} \\
 &  &  \\
\textsc{{\textsc{linear}}} & Prediction $(t|t-1)$ & \num[print-zero-exponent = true,print-implicit-plus=true,print-exponent-implicit-plus=true]{2.137e+01} \ensuremath{\pm} \num[print-zero-exponent = true,print-exponent-implicit-plus=true]{2.9e-02} \\
 & Filtering $(t|t)$ & \num[print-zero-exponent = true,print-implicit-plus=true,print-exponent-implicit-plus=true]{7.109e+01} \ensuremath{\pm} \num[print-zero-exponent = true,print-exponent-implicit-plus=true]{2.3e+00} \\
 & Smoothing $(t|T)$ & \num[print-zero-exponent = true,print-implicit-plus=true,print-exponent-implicit-plus=true]{7.232e+01} \ensuremath{\pm} \num[print-zero-exponent = true,print-exponent-implicit-plus=true]{2.3e+00} \\
 &  &  \\
\textsc{{\textsc{unscented'95}}} & Prediction $(t|t-1)$ & \num[print-zero-exponent = true,print-implicit-plus=true,print-exponent-implicit-plus=true]{1.623e+01} \ensuremath{\pm} \num[print-zero-exponent = true,print-exponent-implicit-plus=true]{2.3e-02} \\
 & Filtering $(t|t)$ & \num[print-zero-exponent = true,print-implicit-plus=true,print-exponent-implicit-plus=true]{1.256e+01} \ensuremath{\pm} \num[print-zero-exponent = true,print-exponent-implicit-plus=true]{4.3e-02} \\
 & Smoothing $(t|T)$ & \num[print-zero-exponent = true,print-implicit-plus=true,print-exponent-implicit-plus=true]{1.507e+01} \ensuremath{\pm} \num[print-zero-exponent = true,print-exponent-implicit-plus=true]{5.6e-02} \\
 &  &  \\
\textsc{{\textsc{unscented'02}}} & Prediction $(t|t-1)$ & --- \ensuremath{\pm} --- \\
 & Filtering $(t|t)$ & --- \ensuremath{\pm} --- \\
 & Smoothing $(t|T)$ & --- \ensuremath{\pm} --- \\
\bottomrule
\end{tabular}
\end{table}

\begin{table}[htbp!]
\centering
\caption{Coverage at 95\% in the Lorenz system state estimation problem}
\label{tab:lorenz_results_coverage95}
\begin{tabular}{lll}
\toprule
Method & Task & Value $\pm$ Monte Carlo standard error \\
\midrule
\textsc{{\textsc{analytic}}} & Prediction $(t|t-1)$ & \num[print-zero-exponent = true,print-implicit-plus=true,print-exponent-implicit-plus=true]{9.850e-01} \ensuremath{\pm} \num[print-zero-exponent = true,print-exponent-implicit-plus=true]{2.2e-04} \\
 & Filtering $(t|t)$ & \num[print-zero-exponent = true,print-implicit-plus=true,print-exponent-implicit-plus=true]{9.794e-01} \ensuremath{\pm} \num[print-zero-exponent = true,print-exponent-implicit-plus=true]{2.1e-04} \\
 & Smoothing $(t|T)$ & \num[print-zero-exponent = true,print-implicit-plus=true,print-exponent-implicit-plus=true]{9.767e-01} \ensuremath{\pm} \num[print-zero-exponent = true,print-exponent-implicit-plus=true]{2.6e-04} \\
 &  &  \\
\textsc{{\textsc{mean-field}}} & Prediction $(t|t-1)$ & \num[print-zero-exponent = true,print-implicit-plus=true,print-exponent-implicit-plus=true]{7.195e-01} \ensuremath{\pm} \num[print-zero-exponent = true,print-exponent-implicit-plus=true]{1.3e-03} \\
 & Filtering $(t|t)$ & \num[print-zero-exponent = true,print-implicit-plus=true,print-exponent-implicit-plus=true]{7.195e-01} \ensuremath{\pm} \num[print-zero-exponent = true,print-exponent-implicit-plus=true]{1.3e-03} \\
 & Smoothing $(t|T)$ & \num[print-zero-exponent = true,print-implicit-plus=true,print-exponent-implicit-plus=true]{7.195e-01} \ensuremath{\pm} \num[print-zero-exponent = true,print-exponent-implicit-plus=true]{1.3e-03} \\
 &  &  \\
\textsc{{\textsc{linear}}} & Prediction $(t|t-1)$ & \num[print-zero-exponent = true,print-implicit-plus=true,print-exponent-implicit-plus=true]{5.942e-02} \ensuremath{\pm} \num[print-zero-exponent = true,print-exponent-implicit-plus=true]{1.1e-03} \\
 & Filtering $(t|t)$ & \num[print-zero-exponent = true,print-implicit-plus=true,print-exponent-implicit-plus=true]{6.143e-02} \ensuremath{\pm} \num[print-zero-exponent = true,print-exponent-implicit-plus=true]{1.1e-03} \\
 & Smoothing $(t|T)$ & \num[print-zero-exponent = true,print-implicit-plus=true,print-exponent-implicit-plus=true]{1.223e-02} \ensuremath{\pm} \num[print-zero-exponent = true,print-exponent-implicit-plus=true]{5.8e-04} \\
 &  &  \\
\textsc{{\textsc{unscented'95}}} & Prediction $(t|t-1)$ & \num[print-zero-exponent = true,print-implicit-plus=true,print-exponent-implicit-plus=true]{5.050e-01} \ensuremath{\pm} \num[print-zero-exponent = true,print-exponent-implicit-plus=true]{1.3e-03} \\
 & Filtering $(t|t)$ & \num[print-zero-exponent = true,print-implicit-plus=true,print-exponent-implicit-plus=true]{5.346e-01} \ensuremath{\pm} \num[print-zero-exponent = true,print-exponent-implicit-plus=true]{1.3e-03} \\
 & Smoothing $(t|T)$ & \num[print-zero-exponent = true,print-implicit-plus=true,print-exponent-implicit-plus=true]{2.728e-01} \ensuremath{\pm} \num[print-zero-exponent = true,print-exponent-implicit-plus=true]{1.2e-03} \\
 &  &  \\
\textsc{{\textsc{unscented'02}}} & Prediction $(t|t-1)$ & \num[print-zero-exponent = true,print-implicit-plus=true,print-exponent-implicit-plus=true]{9.677e-01} \ensuremath{\pm} \num[print-zero-exponent = true,print-exponent-implicit-plus=true]{2.5e-02} \\
 & Filtering $(t|t)$ & \num[print-zero-exponent = true,print-implicit-plus=true,print-exponent-implicit-plus=true]{9.360e-01} \ensuremath{\pm} \num[print-zero-exponent = true,print-exponent-implicit-plus=true]{2.4e-02} \\
 & Smoothing $(t|T)$ & \num[print-zero-exponent = true,print-implicit-plus=true,print-exponent-implicit-plus=true]{8.491e-01} \ensuremath{\pm} \num[print-zero-exponent = true,print-exponent-implicit-plus=true]{6.3e-02} \\
\bottomrule
\end{tabular}
\end{table}

\begin{table}[htbp!]
\centering
\caption{Volume at 95\% in the Lorenz system state estimation problem}
\label{tab:lorenz_results_volume95}
\begin{tabular}{lll}
\toprule
Method & Task & Value $\pm$ Monte Carlo standard error \\
\midrule
\textsc{{\textsc{analytic}}} & Prediction $(t|t-1)$ & \num[print-zero-exponent = true,print-implicit-plus=true,print-exponent-implicit-plus=true]{4.609e+04} \ensuremath{\pm} \num[print-zero-exponent = true,print-exponent-implicit-plus=true]{1.8e+01} \\
 & Filtering $(t|t)$ & \num[print-zero-exponent = true,print-implicit-plus=true,print-exponent-implicit-plus=true]{5.053e+02} \ensuremath{\pm} \num[print-zero-exponent = true,print-exponent-implicit-plus=true]{1.1e-01} \\
 & Smoothing $(t|T)$ & \num[print-zero-exponent = true,print-implicit-plus=true,print-exponent-implicit-plus=true]{4.823e+02} \ensuremath{\pm} \num[print-zero-exponent = true,print-exponent-implicit-plus=true]{1.2e-01} \\
 &  &  \\
\textsc{{\textsc{mean-field}}} & Prediction $(t|t-1)$ & \num[print-zero-exponent = true,print-implicit-plus=true,print-exponent-implicit-plus=true]{4.984e+04} \ensuremath{\pm} 0 \\
 & Filtering $(t|t)$ & \num[print-zero-exponent = true,print-implicit-plus=true,print-exponent-implicit-plus=true]{4.984e+04} \ensuremath{\pm} 0 \\
 & Smoothing $(t|T)$ & \num[print-zero-exponent = true,print-implicit-plus=true,print-exponent-implicit-plus=true]{4.984e+04} \ensuremath{\pm} 0 \\
 &  &  \\
\textsc{{\textsc{linear}}} & Prediction $(t|t-1)$ & \num[print-zero-exponent = true,print-implicit-plus=true,print-exponent-implicit-plus=true]{7.236e+03} \ensuremath{\pm} \num[print-zero-exponent = true,print-exponent-implicit-plus=true]{1.1e+02} \\
 & Filtering $(t|t)$ & \num[print-zero-exponent = true,print-implicit-plus=true,print-exponent-implicit-plus=true]{4.040e+01} \ensuremath{\pm} \num[print-zero-exponent = true,print-exponent-implicit-plus=true]{2.5e-01} \\
 & Smoothing $(t|T)$ & \num[print-zero-exponent = true,print-implicit-plus=true,print-exponent-implicit-plus=true]{1.644e-01} \ensuremath{\pm} \num[print-zero-exponent = true,print-exponent-implicit-plus=true]{9.3e-04} \\
 &  &  \\
\textsc{{\textsc{unscented'95}}} & Prediction $(t|t-1)$ & \num[print-zero-exponent = true,print-implicit-plus=true,print-exponent-implicit-plus=true]{9.029e+03} \ensuremath{\pm} \num[print-zero-exponent = true,print-exponent-implicit-plus=true]{3.1e+01} \\
 & Filtering $(t|t)$ & \num[print-zero-exponent = true,print-implicit-plus=true,print-exponent-implicit-plus=true]{1.349e+02} \ensuremath{\pm} \num[print-zero-exponent = true,print-exponent-implicit-plus=true]{4.1e-01} \\
 & Smoothing $(t|T)$ & \num[print-zero-exponent = true,print-implicit-plus=true,print-exponent-implicit-plus=true]{5.061e+01} \ensuremath{\pm} \num[print-zero-exponent = true,print-exponent-implicit-plus=true]{1.8e-01} \\
 &  &  \\
\textsc{{\textsc{unscented'02}}} & Prediction $(t|t-1)$ & --- \ensuremath{\pm} --- \\
 & Filtering $(t|t)$ & --- \ensuremath{\pm} --- \\
 & Smoothing $(t|T)$ & --- \ensuremath{\pm} --- \\
\bottomrule
\end{tabular}
\end{table}

\begin{table}[htbp!]
\centering
\caption{Cross entropy in the Lorenz system state estimation problem}
\label{tab:lorenz_results_cross_entropy}
\begin{tabular}{lll}
\toprule
Method & Task & Value $\pm$ Monte Carlo standard error \\
\midrule
\textsc{{\textsc{analytic}}} & Prediction $(t|t-1)$ & \num[print-zero-exponent = true,print-implicit-plus=true,print-exponent-implicit-plus=true]{1.014e+01} \ensuremath{\pm} \num[print-zero-exponent = true,print-exponent-implicit-plus=true]{1.9e-03} \\
 & Filtering $(t|t)$ & \num[print-zero-exponent = true,print-implicit-plus=true,print-exponent-implicit-plus=true]{5.717e+00} \ensuremath{\pm} \num[print-zero-exponent = true,print-exponent-implicit-plus=true]{1.2e-03} \\
 & Smoothing $(t|T)$ & \num[print-zero-exponent = true,print-implicit-plus=true,print-exponent-implicit-plus=true]{5.667e+00} \ensuremath{\pm} \num[print-zero-exponent = true,print-exponent-implicit-plus=true]{1.4e-03} \\
 &  &  \\
\textsc{{\textsc{mean-field}}} & Prediction $(t|t-1)$ & \num[print-zero-exponent = true,print-implicit-plus=true,print-exponent-implicit-plus=true]{1.223e+01} \ensuremath{\pm} \num[print-zero-exponent = true,print-exponent-implicit-plus=true]{7.5e-03} \\
 & Filtering $(t|t)$ & \num[print-zero-exponent = true,print-implicit-plus=true,print-exponent-implicit-plus=true]{1.223e+01} \ensuremath{\pm} \num[print-zero-exponent = true,print-exponent-implicit-plus=true]{7.5e-03} \\
 & Smoothing $(t|T)$ & \num[print-zero-exponent = true,print-implicit-plus=true,print-exponent-implicit-plus=true]{1.223e+01} \ensuremath{\pm} \num[print-zero-exponent = true,print-exponent-implicit-plus=true]{7.5e-03} \\
 &  &  \\
\textsc{{\textsc{linear}}} & Prediction $(t|t-1)$ & \num[print-zero-exponent = true,print-implicit-plus=true,print-exponent-implicit-plus=true]{3.088e+02} \ensuremath{\pm} \num[print-zero-exponent = true,print-exponent-implicit-plus=true]{1.2e+00} \\
 & Filtering $(t|t)$ & \num[print-zero-exponent = true,print-implicit-plus=true,print-exponent-implicit-plus=true]{2.945e+02} \ensuremath{\pm} \num[print-zero-exponent = true,print-exponent-implicit-plus=true]{1.2e+00} \\
 & Smoothing $(t|T)$ & \num[print-zero-exponent = true,print-implicit-plus=true,print-exponent-implicit-plus=true]{1.325e+06} \ensuremath{\pm} \num[print-zero-exponent = true,print-exponent-implicit-plus=true]{1.2e+05} \\
 &  &  \\
\textsc{{\textsc{unscented'95}}} & Prediction $(t|t-1)$ & \num[print-zero-exponent = true,print-implicit-plus=true,print-exponent-implicit-plus=true]{2.391e+01} \ensuremath{\pm} \num[print-zero-exponent = true,print-exponent-implicit-plus=true]{1.3e-01} \\
 & Filtering $(t|t)$ & \num[print-zero-exponent = true,print-implicit-plus=true,print-exponent-implicit-plus=true]{1.824e+01} \ensuremath{\pm} \num[print-zero-exponent = true,print-exponent-implicit-plus=true]{1.3e-01} \\
 & Smoothing $(t|T)$ & \num[print-zero-exponent = true,print-implicit-plus=true,print-exponent-implicit-plus=true]{9.860e+01} \ensuremath{\pm} \num[print-zero-exponent = true,print-exponent-implicit-plus=true]{1.7e+00} \\
 &  &  \\
\textsc{{\textsc{unscented'02}}} & Prediction $(t|t-1)$ & --- \ensuremath{\pm} --- \\
 & Filtering $(t|t)$ & --- \ensuremath{\pm} --- \\
 & Smoothing $(t|T)$ & --- \ensuremath{\pm} --- \\
\bottomrule
\end{tabular}
\end{table}

\begin{figure}[H]
\begin{center}
\includegraphics{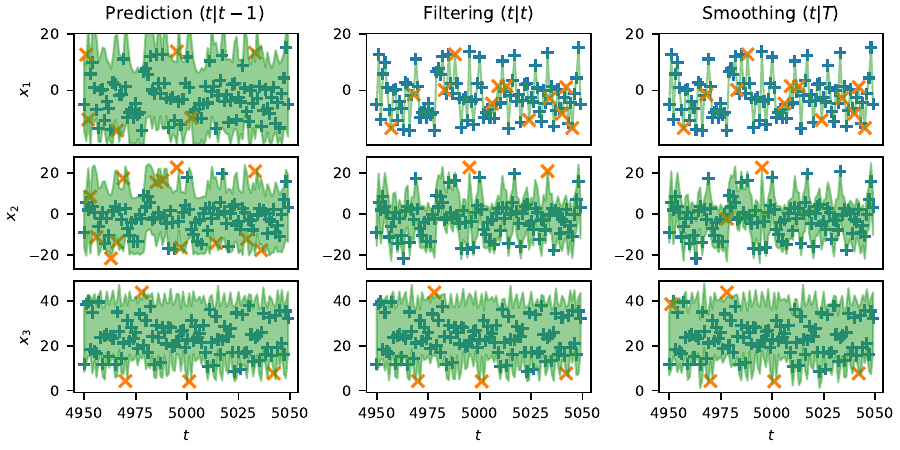}
\end{center}
\caption{\label{fig:Method.ANALYTIC-trajectory}Trajectory excerpt for Kalman filter \textsc{{\textsc{analytic}}} in the Lorenz system state estimation problem. Plus sign indicates hit; cross indicates miss: 10 crosses is best for nominal coverage.}
\end{figure}
\begin{figure}[H]
\begin{center}
\includegraphics{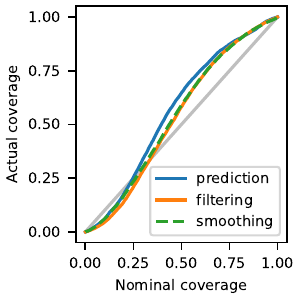}
\end{center}
\caption{\label{fig:Method.ANALYTIC-coverage}Coverage for Kalman filter \textsc{{\textsc{analytic}}} in the Lorenz system state estimation problem. Closer to identity is best.}
\end{figure}
\clearpage
\begin{figure}[H]
\begin{center}
\includegraphics{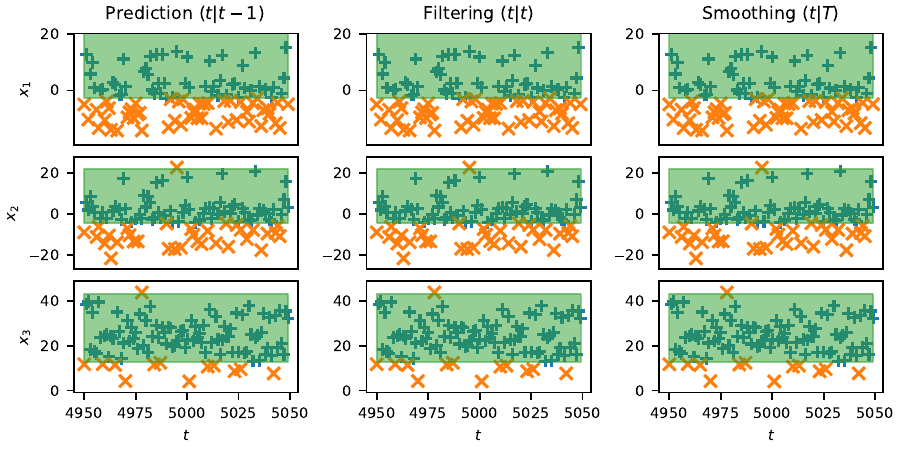}
\end{center}
\caption{\label{fig:Method.MEAN_FIELD-trajectory}Trajectory excerpt for Kalman filter \textsc{{\textsc{mean-field}}} in the Lorenz system state estimation problem. Plus sign indicates hit; cross indicates miss: 10 crosses is best for nominal coverage.}
\end{figure}
\begin{figure}[H]
\begin{center}
\includegraphics{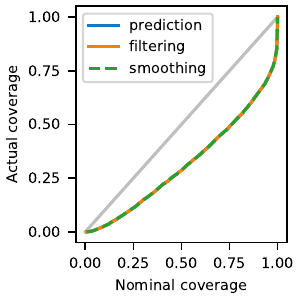}
\end{center}
\caption{\label{fig:Method.MEAN_FIELD-coverage}Coverage for Kalman filter \textsc{{\textsc{mean-field}}} in the Lorenz system state estimation problem. Closer to identity is best.}
\end{figure}
\clearpage
\begin{figure}[H]
\begin{center}
\includegraphics{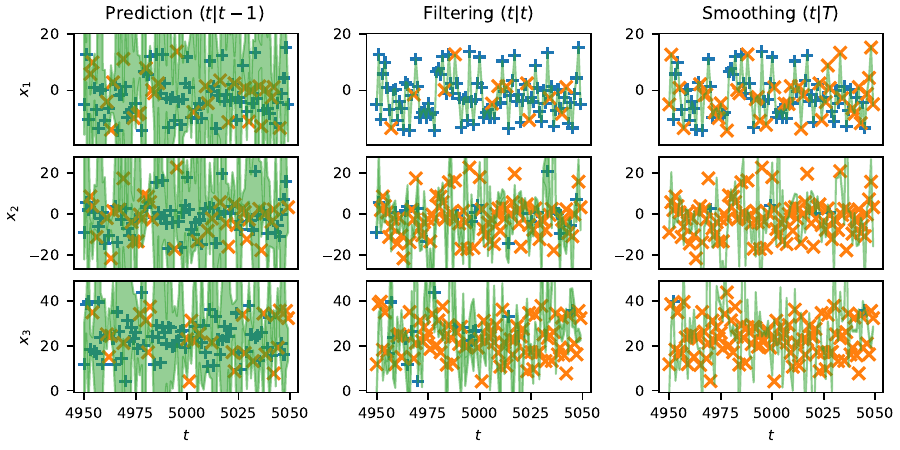}
\end{center}
\caption{\label{fig:Method.LINEAR-trajectory}Trajectory excerpt for Kalman filter \textsc{{\textsc{linear}}} in the Lorenz system state estimation problem. Plus sign indicates hit; cross indicates miss: 10 crosses is best for nominal coverage.}
\end{figure}
\begin{figure}[H]
\begin{center}
\includegraphics{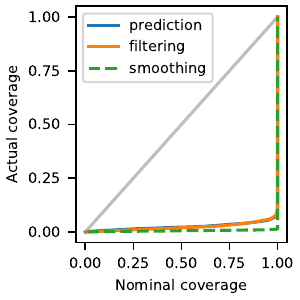}
\end{center}
\caption{\label{fig:Method.LINEAR-coverage}Coverage for Kalman filter \textsc{{\textsc{linear}}} in the Lorenz system state estimation problem. Closer to identity is best.}
\end{figure}
\clearpage
\begin{figure}[H]
\begin{center}
\includegraphics{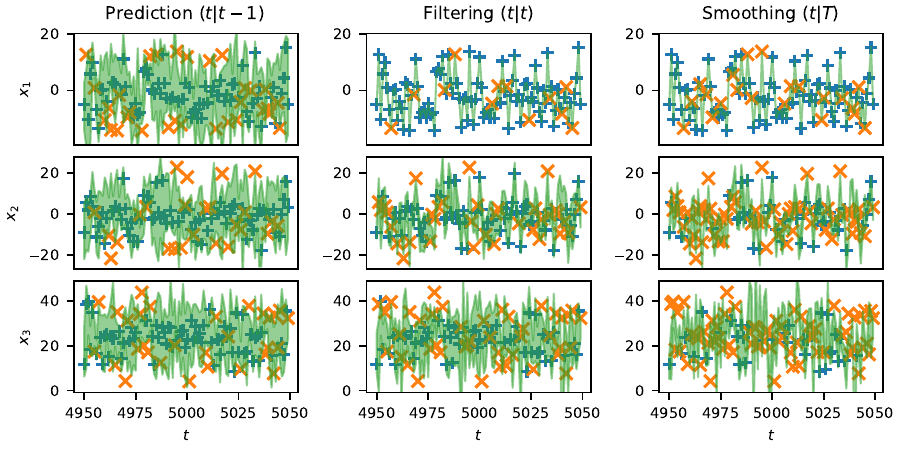}
\end{center}
\caption{\label{fig:Method.UNSCENTED0-trajectory}Trajectory excerpt for Kalman filter \textsc{{\textsc{unscented'95}}} in the Lorenz system state estimation problem. Plus sign indicates hit; cross indicates miss: 10 crosses is best for nominal coverage.}
\end{figure}
\begin{figure}[H]
\begin{center}
\includegraphics{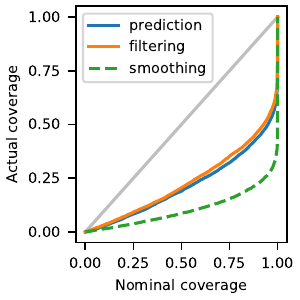}
\end{center}
\caption{\label{fig:Method.UNSCENTED0-coverage}Coverage for Kalman filter \textsc{{\textsc{unscented'95}}} in the Lorenz system state estimation problem. Closer to identity is best.}
\end{figure}
\clearpage
\begin{figure}[H]
\begin{center}
\includegraphics{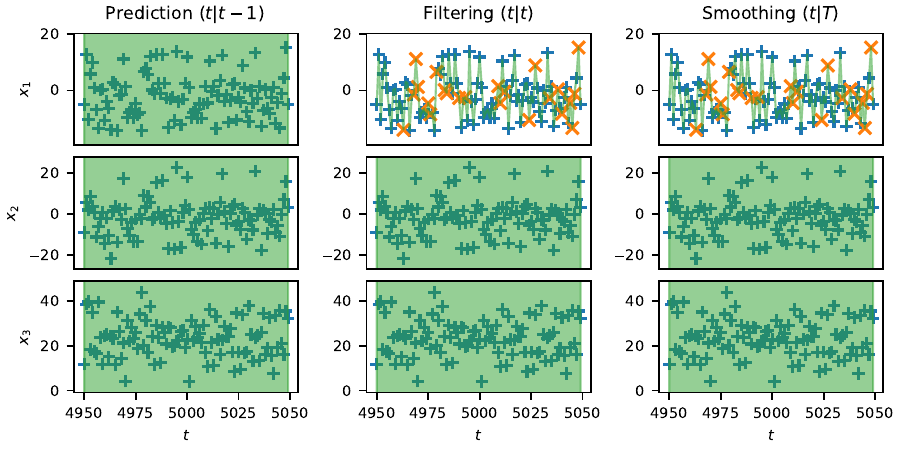}
\end{center}
\caption{\label{fig:Method.UNSCENTED1-trajectory}Trajectory excerpt for Kalman filter \textsc{{\textsc{unscented'02}}} in the Lorenz system state estimation problem. Plus sign indicates hit; cross indicates miss: 10 crosses is best for nominal coverage.}
\end{figure}
\begin{figure}[H]
\begin{center}
\includegraphics{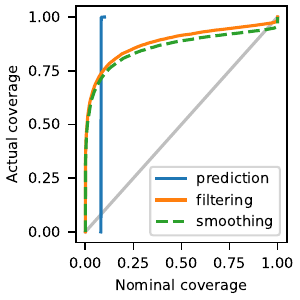}
\end{center}
\caption{\label{fig:Method.UNSCENTED1-coverage}Coverage for Kalman filter \textsc{{\textsc{unscented'02}}} in the Lorenz system state estimation problem. Closer to identity is best.}
\end{figure}
\clearpage

\clearpage
\section{Full results for \S\ref{sec:lti-estimation}}
\label{app:lti-estimation-results}
\begin{table}[htbp!]
\centering
\caption{RMSE for different methods and tasks in the LTI state estimation problem.}
\label{tab:results_rmse}
\begin{tabular}{lll}
\toprule
Method & Task & Value $\pm$ Monte Carlo standard error \\
\midrule
\textsc{analytic} & Prediction $(t|t-1)$ & \num[print-zero-exponent = true,print-implicit-plus=true,print-exponent-implicit-plus=true]{1.377555e+00} \ensuremath{\pm} \num[print-zero-exponent = true,print-exponent-implicit-plus=true]{3.3e-02} \\
 & Filtering $(t|t)$ & \num[print-zero-exponent = true,print-implicit-plus=true,print-exponent-implicit-plus=true]{1.310450e+00} \ensuremath{\pm} \num[print-zero-exponent = true,print-exponent-implicit-plus=true]{3.3e-02} \\
 & Smoothing $(t|T)$ & \num[print-zero-exponent = true,print-implicit-plus=true,print-exponent-implicit-plus=true]{9.936777e-01} \ensuremath{\pm} \num[print-zero-exponent = true,print-exponent-implicit-plus=true]{3.7e-02} \\
 &  &  \\
\textsc{analytic (recal)} & Prediction $(t|t-1)$ & \num[print-zero-exponent = true,print-implicit-plus=true,print-exponent-implicit-plus=true]{1.311325e+00} \ensuremath{\pm} \num[print-zero-exponent = true,print-exponent-implicit-plus=true]{8.5e-03} \\
 & Filtering $(t|t)$ & \num[print-zero-exponent = true,print-implicit-plus=true,print-exponent-implicit-plus=true]{1.241469e+00} \ensuremath{\pm} \num[print-zero-exponent = true,print-exponent-implicit-plus=true]{8.5e-03} \\
 & Smoothing $(t|T)$ & \num[print-zero-exponent = true,print-implicit-plus=true,print-exponent-implicit-plus=true]{8.996923e-01} \ensuremath{\pm} \num[print-zero-exponent = true,print-exponent-implicit-plus=true]{5.3e-03} \\
 &  &  \\
\textsc{linear} & Prediction $(t|t-1)$ & \num[print-zero-exponent = true,print-implicit-plus=true,print-exponent-implicit-plus=true]{3.145826e+01} \ensuremath{\pm} \num[print-zero-exponent = true,print-exponent-implicit-plus=true]{8.0e-01} \\
 & Filtering $(t|t)$ & \num[print-zero-exponent = true,print-implicit-plus=true,print-exponent-implicit-plus=true]{3.152544e+01} \ensuremath{\pm} \num[print-zero-exponent = true,print-exponent-implicit-plus=true]{8.0e-01} \\
 & Smoothing $(t|T)$ & \num[print-zero-exponent = true,print-implicit-plus=true,print-exponent-implicit-plus=true]{3.203451e+01} \ensuremath{\pm} \num[print-zero-exponent = true,print-exponent-implicit-plus=true]{8.1e-01} \\
 &  &  \\
\textsc{linear (recal)} & Prediction $(t|t-1)$ & \num[print-zero-exponent = true,print-implicit-plus=true,print-exponent-implicit-plus=true]{9.831827e+00} \ensuremath{\pm} \num[print-zero-exponent = true,print-exponent-implicit-plus=true]{7.0e-01} \\
 & Filtering $(t|t)$ & \num[print-zero-exponent = true,print-implicit-plus=true,print-exponent-implicit-plus=true]{9.901705e+00} \ensuremath{\pm} \num[print-zero-exponent = true,print-exponent-implicit-plus=true]{7.0e-01} \\
 & Smoothing $(t|T)$ & \num[print-zero-exponent = true,print-implicit-plus=true,print-exponent-implicit-plus=true]{1.085636e+01} \ensuremath{\pm} \num[print-zero-exponent = true,print-exponent-implicit-plus=true]{7.3e-01} \\
 &  &  \\
\textsc{unscented'95} & Prediction $(t|t-1)$ & \num[print-zero-exponent = true,print-implicit-plus=true,print-exponent-implicit-plus=true]{6.560435e+00} \ensuremath{\pm} \num[print-zero-exponent = true,print-exponent-implicit-plus=true]{6.8e-01} \\
 & Filtering $(t|t)$ & \num[print-zero-exponent = true,print-implicit-plus=true,print-exponent-implicit-plus=true]{6.541330e+00} \ensuremath{\pm} \num[print-zero-exponent = true,print-exponent-implicit-plus=true]{6.8e-01} \\
 & Smoothing $(t|T)$ & \num[print-zero-exponent = true,print-implicit-plus=true,print-exponent-implicit-plus=true]{6.613449e+00} \ensuremath{\pm} \num[print-zero-exponent = true,print-exponent-implicit-plus=true]{6.8e-01} \\
 &  &  \\
\textsc{unscented'95 (recal)} & Prediction $(t|t-1)$ & \num[print-zero-exponent = true,print-implicit-plus=true,print-exponent-implicit-plus=true]{1.676370e+00} \ensuremath{\pm} \num[print-zero-exponent = true,print-exponent-implicit-plus=true]{5.5e-02} \\
 & Filtering $(t|t)$ & \num[print-zero-exponent = true,print-implicit-plus=true,print-exponent-implicit-plus=true]{1.612820e+00} \ensuremath{\pm} \num[print-zero-exponent = true,print-exponent-implicit-plus=true]{5.7e-02} \\
 & Smoothing $(t|T)$ & \num[print-zero-exponent = true,print-implicit-plus=true,print-exponent-implicit-plus=true]{1.316173e+00} \ensuremath{\pm} \num[print-zero-exponent = true,print-exponent-implicit-plus=true]{6.9e-02} \\
 &  &  \\
\textsc{unscented'02} & Prediction $(t|t-1)$ & \num[print-zero-exponent = true,print-implicit-plus=true,print-exponent-implicit-plus=true]{2.439038e+01} \ensuremath{\pm} \num[print-zero-exponent = true,print-exponent-implicit-plus=true]{7.6e-01} \\
 & Filtering $(t|t)$ & \num[print-zero-exponent = true,print-implicit-plus=true,print-exponent-implicit-plus=true]{2.457422e+01} \ensuremath{\pm} \num[print-zero-exponent = true,print-exponent-implicit-plus=true]{7.7e-01} \\
 & Smoothing $(t|T)$ & \num[print-zero-exponent = true,print-implicit-plus=true,print-exponent-implicit-plus=true]{2.558191e+01} \ensuremath{\pm} \num[print-zero-exponent = true,print-exponent-implicit-plus=true]{8.1e-01} \\
 &  &  \\
\textsc{unscented'02 (recal)} & Prediction $(t|t-1)$ & \num[print-zero-exponent = true,print-implicit-plus=true,print-exponent-implicit-plus=true]{2.977201e+00} \ensuremath{\pm} \num[print-zero-exponent = true,print-exponent-implicit-plus=true]{1.2e-01} \\
 & Filtering $(t|t)$ & \num[print-zero-exponent = true,print-implicit-plus=true,print-exponent-implicit-plus=true]{2.979570e+00} \ensuremath{\pm} \num[print-zero-exponent = true,print-exponent-implicit-plus=true]{1.3e-01} \\
 & Smoothing $(t|T)$ & \num[print-zero-exponent = true,print-implicit-plus=true,print-exponent-implicit-plus=true]{2.847319e+00} \ensuremath{\pm} \num[print-zero-exponent = true,print-exponent-implicit-plus=true]{1.6e-01} \\
\bottomrule
\end{tabular}
\end{table}

\begin{table}[htbp!]
\centering
\caption{Coverage at 95\% for different methods and tasks in the LTI state estimation problem.}
\label{tab:results_coverage95}
\begin{tabular}{lll}
\toprule
Method & Task & Value $\pm$ Monte Carlo standard error \\
\midrule
\textsc{analytic} & Prediction $(t|t-1)$ & \num[print-zero-exponent = true,print-implicit-plus=true,print-exponent-implicit-plus=true]{9.432500e-01} \ensuremath{\pm} \num[print-zero-exponent = true,print-exponent-implicit-plus=true]{1.7e-03} \\
 & Filtering $(t|t)$ & \num[print-zero-exponent = true,print-implicit-plus=true,print-exponent-implicit-plus=true]{9.414889e-01} \ensuremath{\pm} \num[print-zero-exponent = true,print-exponent-implicit-plus=true]{1.6e-03} \\
 & Smoothing $(t|T)$ & \num[print-zero-exponent = true,print-implicit-plus=true,print-exponent-implicit-plus=true]{9.379111e-01} \ensuremath{\pm} \num[print-zero-exponent = true,print-exponent-implicit-plus=true]{2.0e-03} \\
 &  &  \\
\textsc{analytic (recal)} & Prediction $(t|t-1)$ & \num[print-zero-exponent = true,print-implicit-plus=true,print-exponent-implicit-plus=true]{9.577944e-01} \ensuremath{\pm} \num[print-zero-exponent = true,print-exponent-implicit-plus=true]{1.0e-03} \\
 & Filtering $(t|t)$ & \num[print-zero-exponent = true,print-implicit-plus=true,print-exponent-implicit-plus=true]{9.601778e-01} \ensuremath{\pm} \num[print-zero-exponent = true,print-exponent-implicit-plus=true]{8.8e-04} \\
 & Smoothing $(t|T)$ & \num[print-zero-exponent = true,print-implicit-plus=true,print-exponent-implicit-plus=true]{9.634222e-01} \ensuremath{\pm} \num[print-zero-exponent = true,print-exponent-implicit-plus=true]{7.3e-04} \\
 &  &  \\
\textsc{linear} & Prediction $(t|t-1)$ & \num[print-zero-exponent = true,print-implicit-plus=true,print-exponent-implicit-plus=true]{2.842667e-01} \ensuremath{\pm} \num[print-zero-exponent = true,print-exponent-implicit-plus=true]{7.1e-03} \\
 & Filtering $(t|t)$ & \num[print-zero-exponent = true,print-implicit-plus=true,print-exponent-implicit-plus=true]{2.583111e-01} \ensuremath{\pm} \num[print-zero-exponent = true,print-exponent-implicit-plus=true]{6.5e-03} \\
 & Smoothing $(t|T)$ & \num[print-zero-exponent = true,print-implicit-plus=true,print-exponent-implicit-plus=true]{1.537611e-01} \ensuremath{\pm} \num[print-zero-exponent = true,print-exponent-implicit-plus=true]{5.0e-03} \\
 &  &  \\
\textsc{linear (recal)} & Prediction $(t|t-1)$ & \num[print-zero-exponent = true,print-implicit-plus=true,print-exponent-implicit-plus=true]{6.818222e-01} \ensuremath{\pm} \num[print-zero-exponent = true,print-exponent-implicit-plus=true]{7.3e-03} \\
 & Filtering $(t|t)$ & \num[print-zero-exponent = true,print-implicit-plus=true,print-exponent-implicit-plus=true]{6.641333e-01} \ensuremath{\pm} \num[print-zero-exponent = true,print-exponent-implicit-plus=true]{7.3e-03} \\
 & Smoothing $(t|T)$ & \num[print-zero-exponent = true,print-implicit-plus=true,print-exponent-implicit-plus=true]{5.557056e-01} \ensuremath{\pm} \num[print-zero-exponent = true,print-exponent-implicit-plus=true]{8.1e-03} \\
 &  &  \\
\textsc{unscented'95} & Prediction $(t|t-1)$ & \num[print-zero-exponent = true,print-implicit-plus=true,print-exponent-implicit-plus=true]{8.226000e-01} \ensuremath{\pm} \num[print-zero-exponent = true,print-exponent-implicit-plus=true]{6.1e-03} \\
 & Filtering $(t|t)$ & \num[print-zero-exponent = true,print-implicit-plus=true,print-exponent-implicit-plus=true]{8.011222e-01} \ensuremath{\pm} \num[print-zero-exponent = true,print-exponent-implicit-plus=true]{6.0e-03} \\
 & Smoothing $(t|T)$ & \num[print-zero-exponent = true,print-implicit-plus=true,print-exponent-implicit-plus=true]{7.365500e-01} \ensuremath{\pm} \num[print-zero-exponent = true,print-exponent-implicit-plus=true]{7.4e-03} \\
 &  &  \\
\textsc{unscented'95 (recal)} & Prediction $(t|t-1)$ & \num[print-zero-exponent = true,print-implicit-plus=true,print-exponent-implicit-plus=true]{9.334389e-01} \ensuremath{\pm} \num[print-zero-exponent = true,print-exponent-implicit-plus=true]{2.2e-03} \\
 & Filtering $(t|t)$ & \num[print-zero-exponent = true,print-implicit-plus=true,print-exponent-implicit-plus=true]{9.295889e-01} \ensuremath{\pm} \num[print-zero-exponent = true,print-exponent-implicit-plus=true]{2.0e-03} \\
 & Smoothing $(t|T)$ & \num[print-zero-exponent = true,print-implicit-plus=true,print-exponent-implicit-plus=true]{9.159167e-01} \ensuremath{\pm} \num[print-zero-exponent = true,print-exponent-implicit-plus=true]{2.3e-03} \\
 &  &  \\
\textsc{unscented'02} & Prediction $(t|t-1)$ & \num[print-zero-exponent = true,print-implicit-plus=true,print-exponent-implicit-plus=true]{5.764778e-01} \ensuremath{\pm} \num[print-zero-exponent = true,print-exponent-implicit-plus=true]{8.1e-03} \\
 & Filtering $(t|t)$ & \num[print-zero-exponent = true,print-implicit-plus=true,print-exponent-implicit-plus=true]{5.592444e-01} \ensuremath{\pm} \num[print-zero-exponent = true,print-exponent-implicit-plus=true]{7.9e-03} \\
 & Smoothing $(t|T)$ & \num[print-zero-exponent = true,print-implicit-plus=true,print-exponent-implicit-plus=true]{4.547722e-01} \ensuremath{\pm} \num[print-zero-exponent = true,print-exponent-implicit-plus=true]{9.2e-03} \\
 &  &  \\
\textsc{unscented'02 (recal)} & Prediction $(t|t-1)$ & \num[print-zero-exponent = true,print-implicit-plus=true,print-exponent-implicit-plus=true]{9.097889e-01} \ensuremath{\pm} \num[print-zero-exponent = true,print-exponent-implicit-plus=true]{2.5e-03} \\
 & Filtering $(t|t)$ & \num[print-zero-exponent = true,print-implicit-plus=true,print-exponent-implicit-plus=true]{9.071000e-01} \ensuremath{\pm} \num[print-zero-exponent = true,print-exponent-implicit-plus=true]{2.6e-03} \\
 & Smoothing $(t|T)$ & \num[print-zero-exponent = true,print-implicit-plus=true,print-exponent-implicit-plus=true]{8.839278e-01} \ensuremath{\pm} \num[print-zero-exponent = true,print-exponent-implicit-plus=true]{3.9e-03} \\
\bottomrule
\end{tabular}
\end{table}

\begin{table}[htbp!]
\centering
\caption{Volume at 95\% for different methods and tasks in the LTI state estimation problem.}
\label{tab:results_volume95}
\begin{tabular}{lll}
\toprule
Method & Task & Value $\pm$ Monte Carlo standard error \\
\midrule
\textsc{analytic} & Prediction $(t|t-1)$ & \num[print-zero-exponent = true,print-implicit-plus=true,print-exponent-implicit-plus=true]{1.616272e-01} \ensuremath{\pm} \num[print-zero-exponent = true,print-exponent-implicit-plus=true]{2.7e-04} \\
 & Filtering $(t|t)$ & \num[print-zero-exponent = true,print-implicit-plus=true,print-exponent-implicit-plus=true]{1.166921e-01} \ensuremath{\pm} \num[print-zero-exponent = true,print-exponent-implicit-plus=true]{2.5e-04} \\
 & Smoothing $(t|T)$ & \num[print-zero-exponent = true,print-implicit-plus=true,print-exponent-implicit-plus=true]{2.748990e-02} \ensuremath{\pm} \num[print-zero-exponent = true,print-exponent-implicit-plus=true]{8.1e-05} \\
 &  &  \\
\textsc{analytic (recal)} & Prediction $(t|t-1)$ & \num[print-zero-exponent = true,print-implicit-plus=true,print-exponent-implicit-plus=true]{1.831996e-01} \ensuremath{\pm} \num[print-zero-exponent = true,print-exponent-implicit-plus=true]{2.6e-04} \\
 & Filtering $(t|t)$ & \num[print-zero-exponent = true,print-implicit-plus=true,print-exponent-implicit-plus=true]{1.363137e-01} \ensuremath{\pm} \num[print-zero-exponent = true,print-exponent-implicit-plus=true]{2.3e-04} \\
 & Smoothing $(t|T)$ & \num[print-zero-exponent = true,print-implicit-plus=true,print-exponent-implicit-plus=true]{3.361851e-02} \ensuremath{\pm} \num[print-zero-exponent = true,print-exponent-implicit-plus=true]{7.5e-05} \\
 &  &  \\
\textsc{linear} & Prediction $(t|t-1)$ & \num[print-zero-exponent = true,print-implicit-plus=true,print-exponent-implicit-plus=true]{1.233095e-01} \ensuremath{\pm} \num[print-zero-exponent = true,print-exponent-implicit-plus=true]{4.8e-04} \\
 & Filtering $(t|t)$ & \num[print-zero-exponent = true,print-implicit-plus=true,print-exponent-implicit-plus=true]{8.549655e-02} \ensuremath{\pm} \num[print-zero-exponent = true,print-exponent-implicit-plus=true]{3.4e-04} \\
 & Smoothing $(t|T)$ & \num[print-zero-exponent = true,print-implicit-plus=true,print-exponent-implicit-plus=true]{1.737160e-02} \ensuremath{\pm} \num[print-zero-exponent = true,print-exponent-implicit-plus=true]{9.4e-05} \\
 &  &  \\
\textsc{linear (recal)} & Prediction $(t|t-1)$ & \num[print-zero-exponent = true,print-implicit-plus=true,print-exponent-implicit-plus=true]{2.571232e-01} \ensuremath{\pm} \num[print-zero-exponent = true,print-exponent-implicit-plus=true]{1.5e-03} \\
 & Filtering $(t|t)$ & \num[print-zero-exponent = true,print-implicit-plus=true,print-exponent-implicit-plus=true]{2.026725e-01} \ensuremath{\pm} \num[print-zero-exponent = true,print-exponent-implicit-plus=true]{1.4e-03} \\
 & Smoothing $(t|T)$ & \num[print-zero-exponent = true,print-implicit-plus=true,print-exponent-implicit-plus=true]{5.929100e-02} \ensuremath{\pm} \num[print-zero-exponent = true,print-exponent-implicit-plus=true]{7.7e-04} \\
 &  &  \\
\textsc{unscented'95} & Prediction $(t|t-1)$ & \num[print-zero-exponent = true,print-implicit-plus=true,print-exponent-implicit-plus=true]{1.389514e-01} \ensuremath{\pm} \num[print-zero-exponent = true,print-exponent-implicit-plus=true]{3.1e-04} \\
 & Filtering $(t|t)$ & \num[print-zero-exponent = true,print-implicit-plus=true,print-exponent-implicit-plus=true]{9.609799e-02} \ensuremath{\pm} \num[print-zero-exponent = true,print-exponent-implicit-plus=true]{2.5e-04} \\
 & Smoothing $(t|T)$ & \num[print-zero-exponent = true,print-implicit-plus=true,print-exponent-implicit-plus=true]{1.977639e-02} \ensuremath{\pm} \num[print-zero-exponent = true,print-exponent-implicit-plus=true]{6.5e-05} \\
 &  &  \\
\textsc{unscented'95 (recal)} & Prediction $(t|t-1)$ & \num[print-zero-exponent = true,print-implicit-plus=true,print-exponent-implicit-plus=true]{1.925837e-01} \ensuremath{\pm} \num[print-zero-exponent = true,print-exponent-implicit-plus=true]{4.7e-04} \\
 & Filtering $(t|t)$ & \num[print-zero-exponent = true,print-implicit-plus=true,print-exponent-implicit-plus=true]{1.435821e-01} \ensuremath{\pm} \num[print-zero-exponent = true,print-exponent-implicit-plus=true]{4.2e-04} \\
 & Smoothing $(t|T)$ & \num[print-zero-exponent = true,print-implicit-plus=true,print-exponent-implicit-plus=true]{3.384204e-02} \ensuremath{\pm} \num[print-zero-exponent = true,print-exponent-implicit-plus=true]{1.7e-04} \\
 &  &  \\
\textsc{unscented'02} & Prediction $(t|t-1)$ & \num[print-zero-exponent = true,print-implicit-plus=true,print-exponent-implicit-plus=true]{1.823639e-01} \ensuremath{\pm} \num[print-zero-exponent = true,print-exponent-implicit-plus=true]{5.3e-04} \\
 & Filtering $(t|t)$ & \num[print-zero-exponent = true,print-implicit-plus=true,print-exponent-implicit-plus=true]{1.371937e-01} \ensuremath{\pm} \num[print-zero-exponent = true,print-exponent-implicit-plus=true]{4.7e-04} \\
 & Smoothing $(t|T)$ & \num[print-zero-exponent = true,print-implicit-plus=true,print-exponent-implicit-plus=true]{3.178623e-02} \ensuremath{\pm} \num[print-zero-exponent = true,print-exponent-implicit-plus=true]{2.1e-04} \\
 &  &  \\
\textsc{unscented'02 (recal)} & Prediction $(t|t-1)$ & \num[print-zero-exponent = true,print-implicit-plus=true,print-exponent-implicit-plus=true]{3.317351e-01} \ensuremath{\pm} \num[print-zero-exponent = true,print-exponent-implicit-plus=true]{1.3e-03} \\
 & Filtering $(t|t)$ & \num[print-zero-exponent = true,print-implicit-plus=true,print-exponent-implicit-plus=true]{2.798133e-01} \ensuremath{\pm} \num[print-zero-exponent = true,print-exponent-implicit-plus=true]{1.4e-03} \\
 & Smoothing $(t|T)$ & \num[print-zero-exponent = true,print-implicit-plus=true,print-exponent-implicit-plus=true]{9.692433e-02} \ensuremath{\pm} \num[print-zero-exponent = true,print-exponent-implicit-plus=true]{9.6e-04} \\
\bottomrule
\end{tabular}
\end{table}

\begin{table}[htbp!]
\centering
\caption{Cross entropy for different methods and tasks in the LTI state estimation problem.}
\label{tab:results_cross_entropy}
\begin{tabular}{lll}
\toprule
Method & Task & Value $\pm$ Monte Carlo standard error \\
\midrule
\textsc{analytic} & Prediction $(t|t-1)$ & \num[print-zero-exponent = true,print-implicit-plus=true,print-exponent-implicit-plus=true]{-1.726823e+00} \ensuremath{\pm} \num[print-zero-exponent = true,print-exponent-implicit-plus=true]{7.9e-01} \\
 & Filtering $(t|t)$ & \num[print-zero-exponent = true,print-implicit-plus=true,print-exponent-implicit-plus=true]{-4.652060e-01} \ensuremath{\pm} \num[print-zero-exponent = true,print-exponent-implicit-plus=true]{2.3e+00} \\
 & Smoothing $(t|T)$ & \num[print-zero-exponent = true,print-implicit-plus=true,print-exponent-implicit-plus=true]{1.083637e-01} \ensuremath{\pm} \num[print-zero-exponent = true,print-exponent-implicit-plus=true]{4.0e+00} \\
 &  &  \\
\textsc{analytic (recal)} & Prediction $(t|t-1)$ & \num[print-zero-exponent = true,print-implicit-plus=true,print-exponent-implicit-plus=true]{-2.725533e+00} \ensuremath{\pm} \num[print-zero-exponent = true,print-exponent-implicit-plus=true]{9.6e-03} \\
 & Filtering $(t|t)$ & \num[print-zero-exponent = true,print-implicit-plus=true,print-exponent-implicit-plus=true]{-3.257829e+00} \ensuremath{\pm} \num[print-zero-exponent = true,print-exponent-implicit-plus=true]{8.5e-03} \\
 & Smoothing $(t|T)$ & \num[print-zero-exponent = true,print-implicit-plus=true,print-exponent-implicit-plus=true]{-4.640158e+00} \ensuremath{\pm} \num[print-zero-exponent = true,print-exponent-implicit-plus=true]{7.2e-03} \\
 &  &  \\
\textsc{linear} & Prediction $(t|t-1)$ & \num[print-zero-exponent = true,print-implicit-plus=true,print-exponent-implicit-plus=true]{2.362394e+04} \ensuremath{\pm} \num[print-zero-exponent = true,print-exponent-implicit-plus=true]{1.2e+03} \\
 & Filtering $(t|t)$ & \num[print-zero-exponent = true,print-implicit-plus=true,print-exponent-implicit-plus=true]{1.794416e+05} \ensuremath{\pm} \num[print-zero-exponent = true,print-exponent-implicit-plus=true]{1.1e+04} \\
 & Smoothing $(t|T)$ & \num[print-zero-exponent = true,print-implicit-plus=true,print-exponent-implicit-plus=true]{2.116272e+05} \ensuremath{\pm} \num[print-zero-exponent = true,print-exponent-implicit-plus=true]{1.2e+04} \\
 &  &  \\
\textsc{linear (recal)} & Prediction $(t|t-1)$ & \num[print-zero-exponent = true,print-implicit-plus=true,print-exponent-implicit-plus=true]{1.027865e+02} \ensuremath{\pm} \num[print-zero-exponent = true,print-exponent-implicit-plus=true]{1.1e+01} \\
 & Filtering $(t|t)$ & \num[print-zero-exponent = true,print-implicit-plus=true,print-exponent-implicit-plus=true]{3.256018e+02} \ensuremath{\pm} \num[print-zero-exponent = true,print-exponent-implicit-plus=true]{7.0e+01} \\
 & Smoothing $(t|T)$ & \num[print-zero-exponent = true,print-implicit-plus=true,print-exponent-implicit-plus=true]{4.992393e+02} \ensuremath{\pm} \num[print-zero-exponent = true,print-exponent-implicit-plus=true]{8.8e+01} \\
 &  &  \\
\textsc{unscented'95} & Prediction $(t|t-1)$ & \num[print-zero-exponent = true,print-implicit-plus=true,print-exponent-implicit-plus=true]{9.175953e+02} \ensuremath{\pm} \num[print-zero-exponent = true,print-exponent-implicit-plus=true]{1.9e+02} \\
 & Filtering $(t|t)$ & \num[print-zero-exponent = true,print-implicit-plus=true,print-exponent-implicit-plus=true]{3.885867e+03} \ensuremath{\pm} \num[print-zero-exponent = true,print-exponent-implicit-plus=true]{8.5e+02} \\
 & Smoothing $(t|T)$ & \num[print-zero-exponent = true,print-implicit-plus=true,print-exponent-implicit-plus=true]{5.232937e+03} \ensuremath{\pm} \num[print-zero-exponent = true,print-exponent-implicit-plus=true]{1.1e+03} \\
 &  &  \\
\textsc{unscented'95 (recal)} & Prediction $(t|t-1)$ & \num[print-zero-exponent = true,print-implicit-plus=true,print-exponent-implicit-plus=true]{-1.694058e+00} \ensuremath{\pm} \num[print-zero-exponent = true,print-exponent-implicit-plus=true]{2.1e-01} \\
 & Filtering $(t|t)$ & \num[print-zero-exponent = true,print-implicit-plus=true,print-exponent-implicit-plus=true]{-1.683487e+00} \ensuremath{\pm} \num[print-zero-exponent = true,print-exponent-implicit-plus=true]{3.6e-01} \\
 & Smoothing $(t|T)$ & \num[print-zero-exponent = true,print-implicit-plus=true,print-exponent-implicit-plus=true]{-2.069659e+00} \ensuremath{\pm} \num[print-zero-exponent = true,print-exponent-implicit-plus=true]{6.4e-01} \\
 &  &  \\
\textsc{unscented'02} & Prediction $(t|t-1)$ & \num[print-zero-exponent = true,print-implicit-plus=true,print-exponent-implicit-plus=true]{1.119020e+04} \ensuremath{\pm} \num[print-zero-exponent = true,print-exponent-implicit-plus=true]{6.4e+02} \\
 & Filtering $(t|t)$ & \num[print-zero-exponent = true,print-implicit-plus=true,print-exponent-implicit-plus=true]{6.514542e+04} \ensuremath{\pm} \num[print-zero-exponent = true,print-exponent-implicit-plus=true]{4.4e+03} \\
 & Smoothing $(t|T)$ & \num[print-zero-exponent = true,print-implicit-plus=true,print-exponent-implicit-plus=true]{8.183052e+04} \ensuremath{\pm} \num[print-zero-exponent = true,print-exponent-implicit-plus=true]{5.2e+03} \\
 &  &  \\
\textsc{unscented'02 (recal)} & Prediction $(t|t-1)$ & \num[print-zero-exponent = true,print-implicit-plus=true,print-exponent-implicit-plus=true]{9.706521e-02} \ensuremath{\pm} \num[print-zero-exponent = true,print-exponent-implicit-plus=true]{2.2e-01} \\
 & Filtering $(t|t)$ & \num[print-zero-exponent = true,print-implicit-plus=true,print-exponent-implicit-plus=true]{1.388761e+00} \ensuremath{\pm} \num[print-zero-exponent = true,print-exponent-implicit-plus=true]{7.4e-01} \\
 & Smoothing $(t|T)$ & \num[print-zero-exponent = true,print-implicit-plus=true,print-exponent-implicit-plus=true]{3.184046e+00} \ensuremath{\pm} \num[print-zero-exponent = true,print-exponent-implicit-plus=true]{1.0e+00} \\
\bottomrule
\end{tabular}
\end{table}

\begin{figure}[H]
\centering
\includegraphics[width=\linewidth]{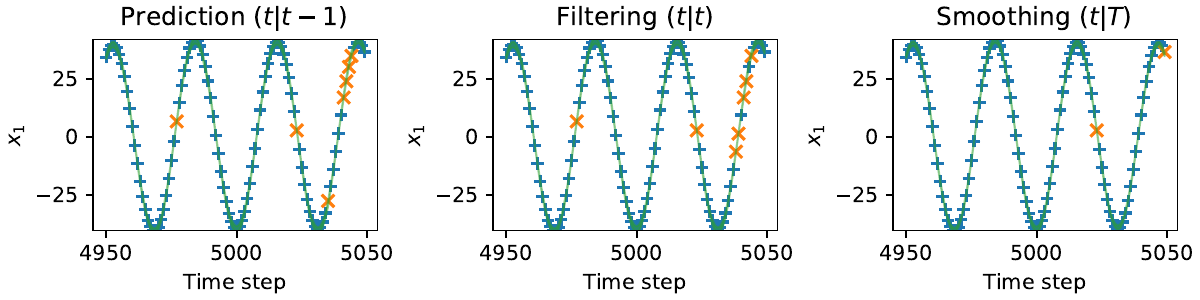}
\caption{\label{fig:trajectory-analytic-no} Trajectory excerpt for Kalman filter \textsc{analytic} in the LTI state estimation problem. Plus sign indicates hit; cross indicates miss. Expect 10 misses for nominal 90\% coverage.}
\end{figure}
\begin{figure}[H]
\centering
\includegraphics[width=\linewidth]{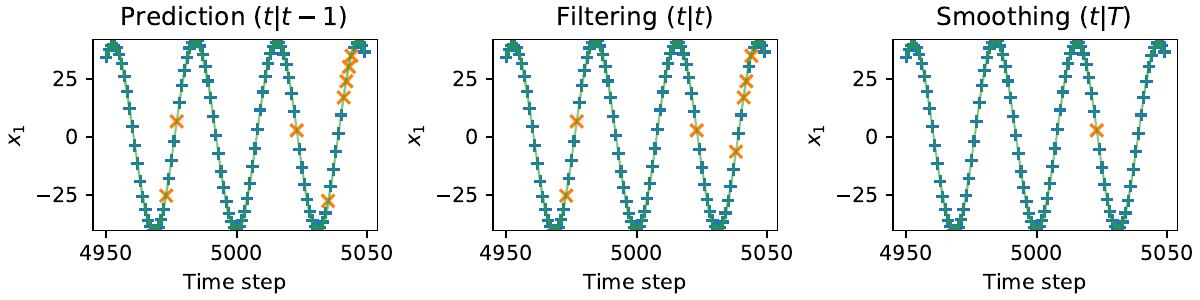}
\caption{\label{fig:trajectory-analytic-yes} Trajectory excerpt for Kalman filter \textsc{analytic (recal)} in the LTI state estimation problem. Plus sign indicates hit; cross indicates miss. Expect 10 misses for nominal 90\% coverage.}
\end{figure}
\begin{figure}[H]
\centering
\includegraphics[width=\linewidth]{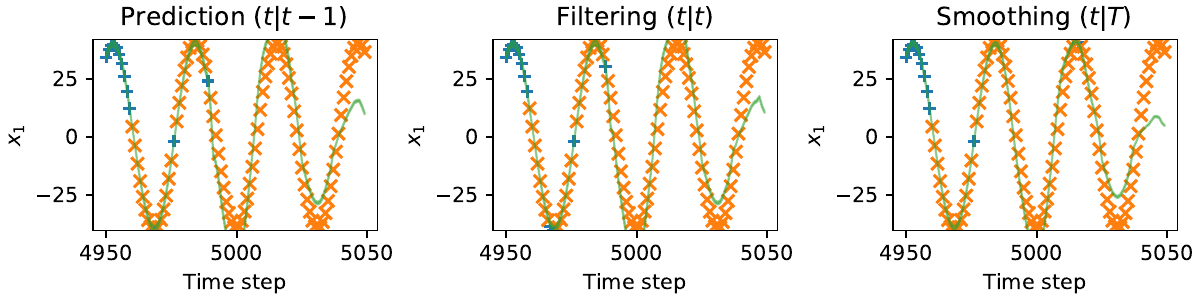}
\caption{\label{fig:trajectory-linear-no} Trajectory excerpt for Kalman filter \textsc{linear} in the LTI state estimation problem. Plus sign indicates hit; cross indicates miss. Expect 10 misses for nominal 90\% coverage.}
\end{figure}
\begin{figure}[H]
\centering
\includegraphics[width=\linewidth]{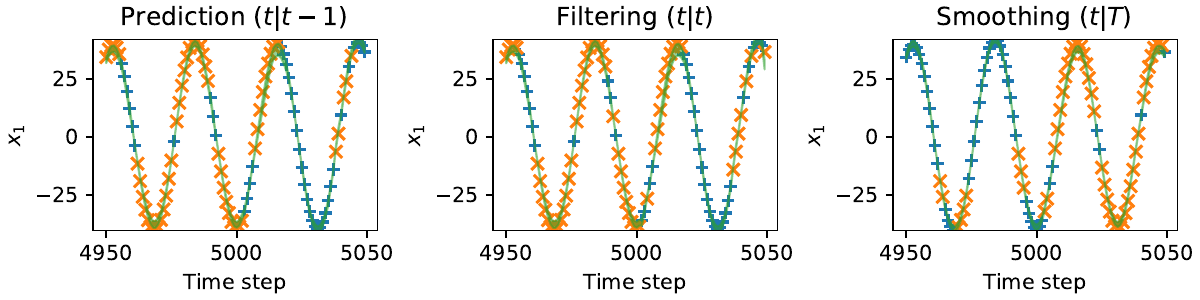}
\caption{\label{fig:trajectory-linear-yes} Trajectory excerpt for Kalman filter \textsc{linear (recal)} in the LTI state estimation problem. Plus sign indicates hit; cross indicates miss. Expect 10 misses for nominal 90\% coverage.}
\end{figure}
\begin{figure}[H]
\centering
\includegraphics[width=\linewidth]{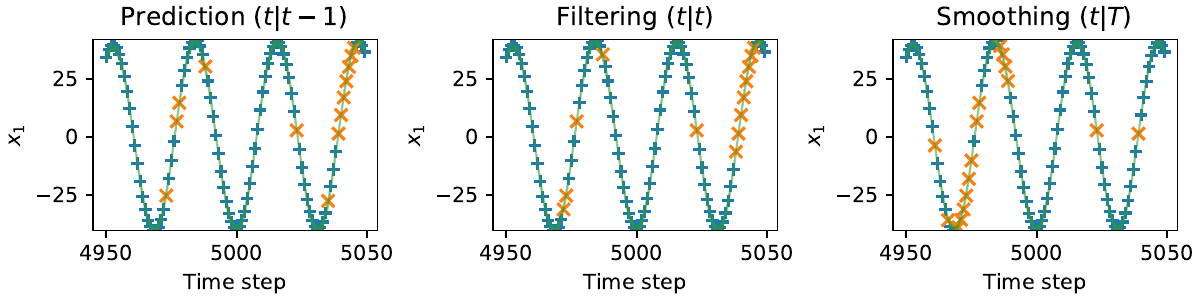}
\caption{\label{fig:trajectory-unscented0-no} Trajectory excerpt for Kalman filter \textsc{unscented'95} in the LTI state estimation problem. Plus sign indicates hit; cross indicates miss. Expect 10 misses for nominal 90\% coverage.}
\end{figure}
\begin{figure}[H]
\centering
\includegraphics[width=\linewidth]{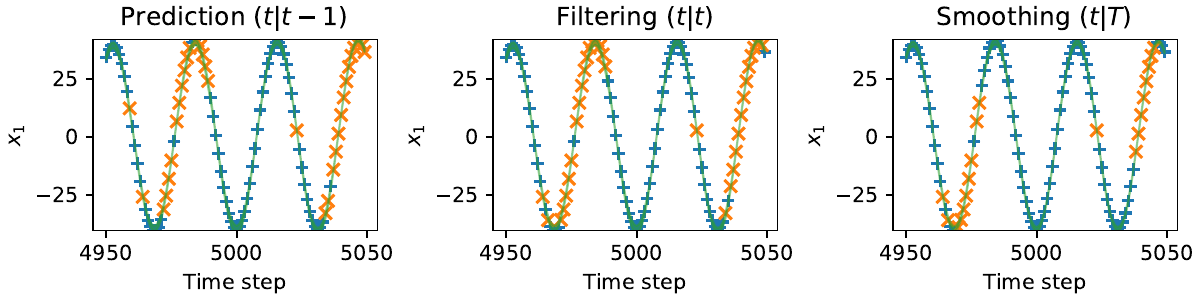}
\caption{\label{fig:trajectory-unscented0-yes} Trajectory excerpt for Kalman filter \textsc{unscented'95 (recal)} in the LTI state estimation problem. Plus sign indicates hit; cross indicates miss. Expect 10 misses for nominal 90\% coverage.}
\end{figure}
\begin{figure}[H]
\centering
\includegraphics[width=\linewidth]{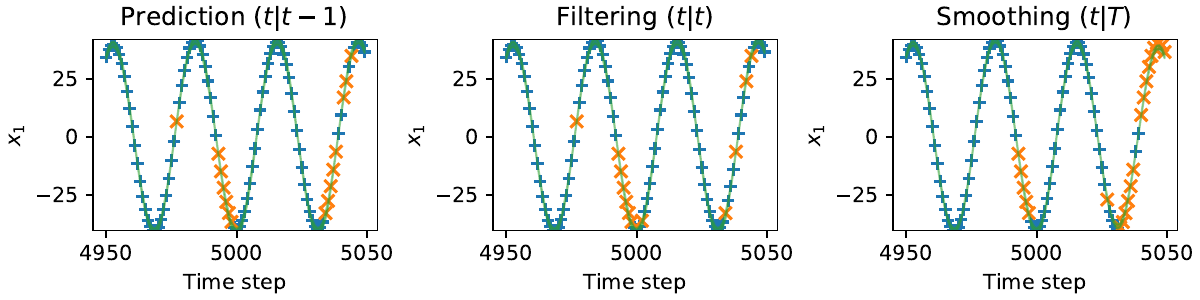}
\caption{\label{fig:trajectory-unscented1-no} Trajectory excerpt for Kalman filter \textsc{unscented'02} in the LTI state estimation problem. Plus sign indicates hit; cross indicates miss. Expect 10 misses for nominal 90\% coverage.}
\end{figure}
\begin{figure}[H]
\centering
\includegraphics[width=\linewidth]{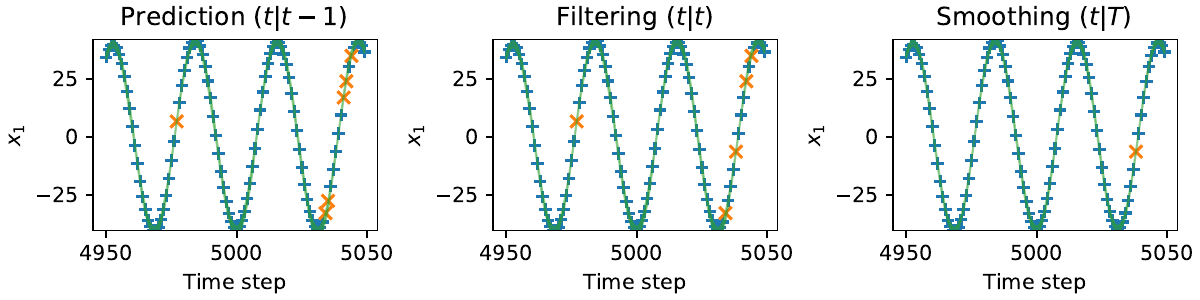}
\caption{\label{fig:trajectory-unscented1-yes} Trajectory excerpt for Kalman filter \textsc{unscented'02 (recal)} in the LTI state estimation problem. Plus sign indicates hit; cross indicates miss. Expect 10 misses for nominal 90\% coverage.}
\end{figure}

\subsection{Effect of recalibration}

\begin{figure}[H]
  \centering
  \includegraphics[width=\linewidth]{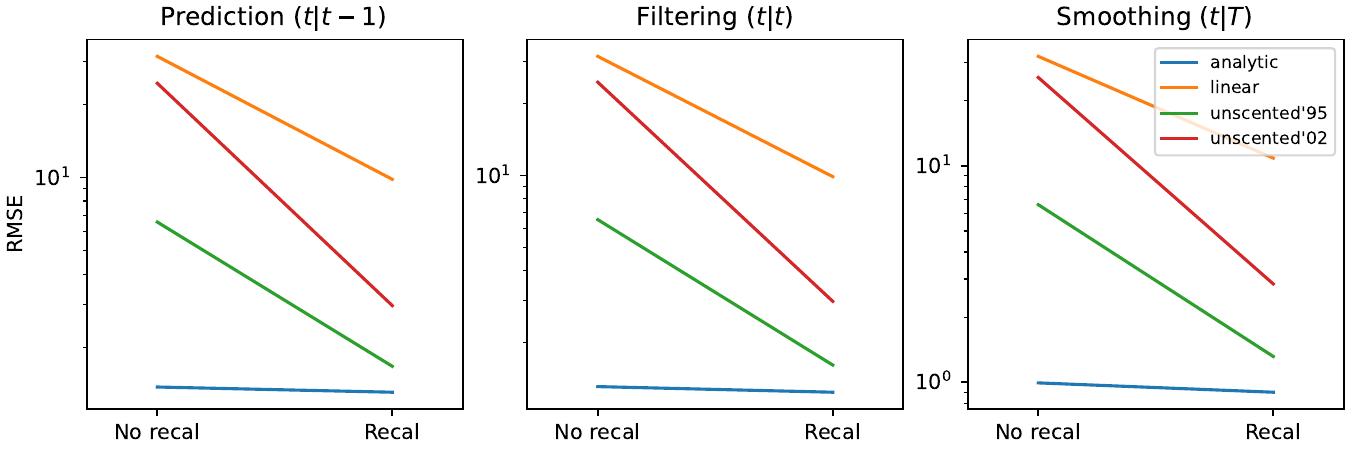}
  \caption{\label{fig:recal-effect-rmse}Effect of recalibration on RMSE for each method and inference task in the LTI state estimation problem. Each line segment connects a method's no-recal value to its recal value; vertical axis is log-scaled.}
\end{figure}

\begin{figure}[H]
  \centering
  \includegraphics[width=\linewidth]{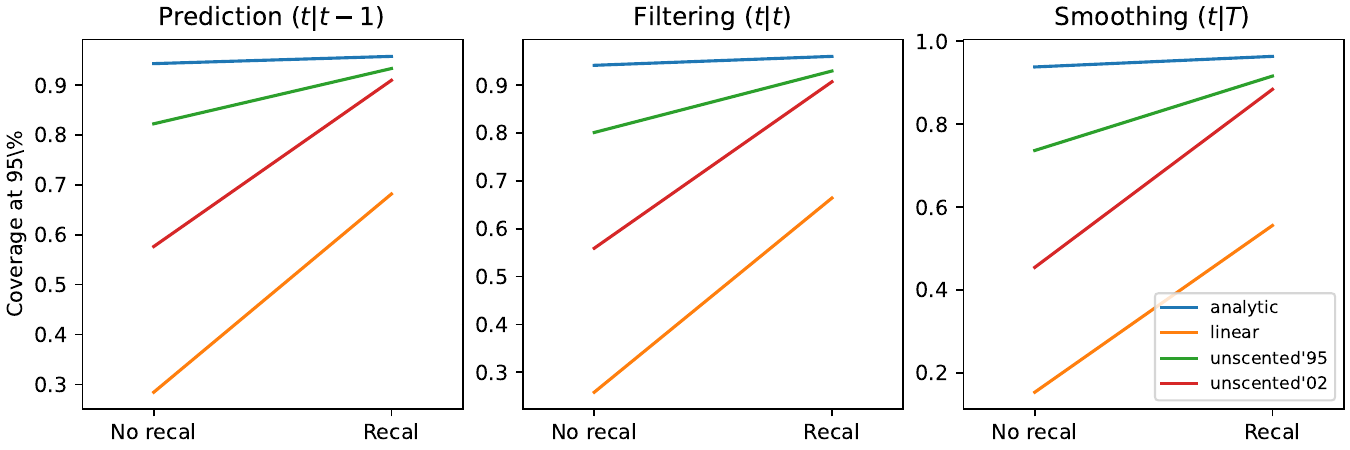}
  \caption{\label{fig:recal-effect-coverage95}Effect of recalibration on 95\% coverage.}
\end{figure}

\begin{figure}[H]
  \centering
  \includegraphics[width=\linewidth]{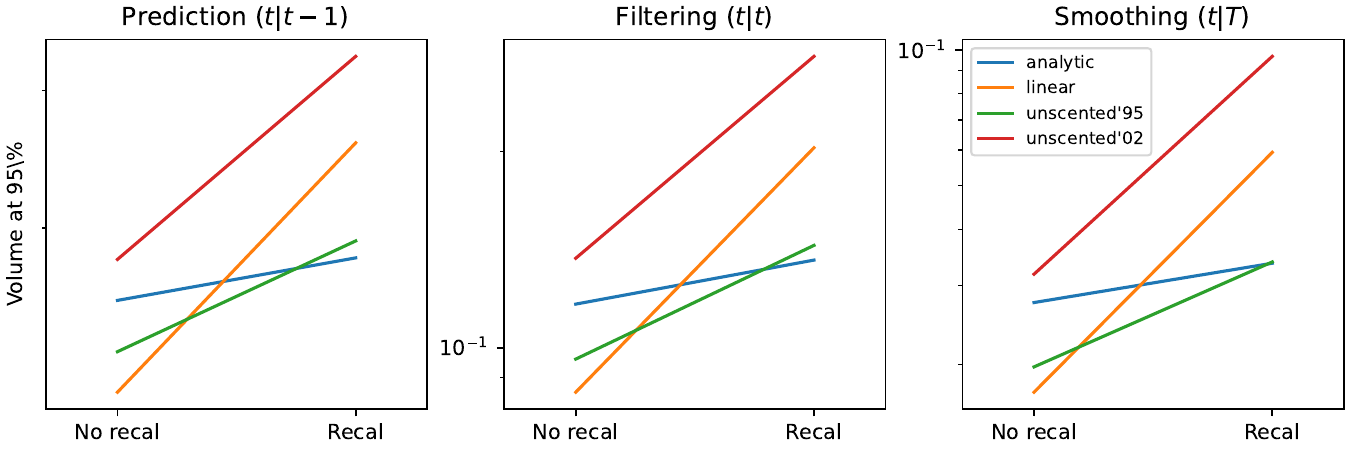}
  \caption{\label{fig:recal-effect-volume95}Effect of recalibration on 95\% confidence volume (log-scaled).}
\end{figure}

\begin{figure}[H]
  \centering
  \includegraphics[width=\linewidth]{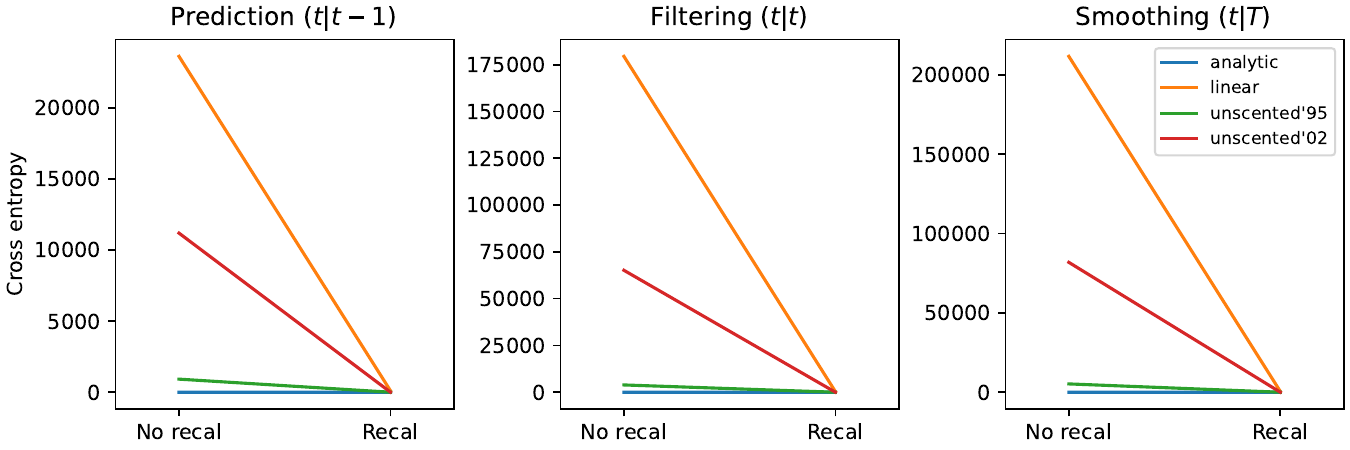}
  \caption{\label{fig:recal-effect-cross-entropy}Effect of recalibration on cross entropy.}
\end{figure}

\clearpage
\section{Full results for \S\ref{sec:lti-regulation}}
\label{app:lti-regulation-results}
\begin{table}[htbp!]
\centering
\caption{LQR performance of different methods and tasks in the LTI regulation problem.}
\label{tab:results_regulation}
\begin{tabular}{lll}
\toprule
Method & Task & Value $\pm$ Monte Carlo standard error \\
\midrule
\textsc{analytic} & State cost / LQR & \num[print-zero-exponent = true,print-implicit-plus=true,print-exponent-implicit-plus=true]{1.053475e+00} \ensuremath{\pm} \num[print-zero-exponent = true,print-exponent-implicit-plus=true]{3.2e-03} \\
 & Control cost / LQR & \num[print-zero-exponent = true,print-implicit-plus=true,print-exponent-implicit-plus=true]{9.999039e-01} \ensuremath{\pm} \num[print-zero-exponent = true,print-exponent-implicit-plus=true]{3.0e-03} \\
 & Total cost / LQR & \num[print-zero-exponent = true,print-implicit-plus=true,print-exponent-implicit-plus=true]{1.046893e+00} \ensuremath{\pm} \num[print-zero-exponent = true,print-exponent-implicit-plus=true]{3.2e-03} \\
 &  &  \\
\textsc{analytic (recal)} & State cost / LQR & \num[print-zero-exponent = true,print-implicit-plus=true,print-exponent-implicit-plus=true]{1.049540e+00} \ensuremath{\pm} \num[print-zero-exponent = true,print-exponent-implicit-plus=true]{3.4e-03} \\
 & Control cost / LQR & \num[print-zero-exponent = true,print-implicit-plus=true,print-exponent-implicit-plus=true]{1.027164e+00} \ensuremath{\pm} \num[print-zero-exponent = true,print-exponent-implicit-plus=true]{3.2e-03} \\
 & Total cost / LQR & \num[print-zero-exponent = true,print-implicit-plus=true,print-exponent-implicit-plus=true]{1.046791e+00} \ensuremath{\pm} \num[print-zero-exponent = true,print-exponent-implicit-plus=true]{3.3e-03} \\
 &  &  \\
\textsc{linear} & State cost / LQR & \num[print-zero-exponent = true,print-implicit-plus=true,print-exponent-implicit-plus=true]{1.914271e+06} \ensuremath{\pm} \num[print-zero-exponent = true,print-exponent-implicit-plus=true]{6.8e+05} \\
 & Control cost / LQR & \num[print-zero-exponent = true,print-implicit-plus=true,print-exponent-implicit-plus=true]{1.687060e+02} \ensuremath{\pm} \num[print-zero-exponent = true,print-exponent-implicit-plus=true]{4.3e+00} \\
 & Total cost / LQR & \num[print-zero-exponent = true,print-implicit-plus=true,print-exponent-implicit-plus=true]{1.679098e+06} \ensuremath{\pm} \num[print-zero-exponent = true,print-exponent-implicit-plus=true]{5.9e+05} \\
 &  &  \\
\textsc{linear (recal)} & State cost / LQR & \num[print-zero-exponent = true,print-implicit-plus=true,print-exponent-implicit-plus=true]{5.340074e+06} \ensuremath{\pm} \num[print-zero-exponent = true,print-exponent-implicit-plus=true]{4.7e+05} \\
 & Control cost / LQR & \num[print-zero-exponent = true,print-implicit-plus=true,print-exponent-implicit-plus=true]{3.149745e+00} \ensuremath{\pm} \num[print-zero-exponent = true,print-exponent-implicit-plus=true]{1.8e-01} \\
 & Total cost / LQR & \num[print-zero-exponent = true,print-implicit-plus=true,print-exponent-implicit-plus=true]{4.683974e+06} \ensuremath{\pm} \num[print-zero-exponent = true,print-exponent-implicit-plus=true]{4.2e+05} \\
 &  &  \\
\textsc{unscented'95} & State cost / LQR & \num[print-zero-exponent = true,print-implicit-plus=true,print-exponent-implicit-plus=true]{2.838426e+05} \ensuremath{\pm} \num[print-zero-exponent = true,print-exponent-implicit-plus=true]{1.4e+05} \\
 & Control cost / LQR & \num[print-zero-exponent = true,print-implicit-plus=true,print-exponent-implicit-plus=true]{1.689836e+02} \ensuremath{\pm} \num[print-zero-exponent = true,print-exponent-implicit-plus=true]{1.1e+01} \\
 & Total cost / LQR & \num[print-zero-exponent = true,print-implicit-plus=true,print-exponent-implicit-plus=true]{2.489894e+05} \ensuremath{\pm} \num[print-zero-exponent = true,print-exponent-implicit-plus=true]{1.3e+05} \\
 &  &  \\
\textsc{unscented'95 (recal)} & State cost / LQR & \num[print-zero-exponent = true,print-implicit-plus=true,print-exponent-implicit-plus=true]{9.005846e+00} \ensuremath{\pm} \num[print-zero-exponent = true,print-exponent-implicit-plus=true]{2.2e+00} \\
 & Control cost / LQR & \num[print-zero-exponent = true,print-implicit-plus=true,print-exponent-implicit-plus=true]{3.704188e+00} \ensuremath{\pm} \num[print-zero-exponent = true,print-exponent-implicit-plus=true]{1.9e-01} \\
 & Total cost / LQR & \num[print-zero-exponent = true,print-implicit-plus=true,print-exponent-implicit-plus=true]{8.354465e+00} \ensuremath{\pm} \num[print-zero-exponent = true,print-exponent-implicit-plus=true]{2.0e+00} \\
 &  &  \\
\textsc{unscented'02} & State cost / LQR & \num[print-zero-exponent = true,print-implicit-plus=true,print-exponent-implicit-plus=true]{5.775930e+05} \ensuremath{\pm} \num[print-zero-exponent = true,print-exponent-implicit-plus=true]{2.6e+05} \\
 & Control cost / LQR & \num[print-zero-exponent = true,print-implicit-plus=true,print-exponent-implicit-plus=true]{3.378205e+02} \ensuremath{\pm} \num[print-zero-exponent = true,print-exponent-implicit-plus=true]{1.2e+01} \\
 & Total cost / LQR & \num[print-zero-exponent = true,print-implicit-plus=true,print-exponent-implicit-plus=true]{5.066693e+05} \ensuremath{\pm} \num[print-zero-exponent = true,print-exponent-implicit-plus=true]{2.3e+05} \\
 &  &  \\
\textsc{unscented'02 (recal)} & State cost / LQR & \num[print-zero-exponent = true,print-implicit-plus=true,print-exponent-implicit-plus=true]{1.693187e+03} \ensuremath{\pm} \num[print-zero-exponent = true,print-exponent-implicit-plus=true]{3.0e+02} \\
 & Control cost / LQR & \num[print-zero-exponent = true,print-implicit-plus=true,print-exponent-implicit-plus=true]{7.595572e-03} \ensuremath{\pm} \num[print-zero-exponent = true,print-exponent-implicit-plus=true]{4.4e-03} \\
 & Total cost / LQR & \num[print-zero-exponent = true,print-implicit-plus=true,print-exponent-implicit-plus=true]{1.485157e+03} \ensuremath{\pm} \num[print-zero-exponent = true,print-exponent-implicit-plus=true]{2.7e+02} \\
\bottomrule
\end{tabular}
\end{table}

\begin{table}[htbp!]
\centering
\caption{RMSE of the Kalman filters in closed loop in the LTI regulation problem.}
\label{tab:results_regulation_rmse}
\begin{tabular}{lll}
\toprule
Method & Task & Value $\pm$ Monte Carlo standard error \\
\midrule
\textsc{analytic} & Filtering & \num[print-zero-exponent = true,print-implicit-plus=true,print-exponent-implicit-plus=true]{3.568368e-02} \ensuremath{\pm} \num[print-zero-exponent = true,print-exponent-implicit-plus=true]{8.3e-05} \\
 &  &  \\
\textsc{analytic (recal)} & Filtering & \num[print-zero-exponent = true,print-implicit-plus=true,print-exponent-implicit-plus=true]{3.590258e-02} \ensuremath{\pm} \num[print-zero-exponent = true,print-exponent-implicit-plus=true]{7.6e-05} \\
 &  &  \\
\textsc{linear} & Filtering & \num[print-zero-exponent = true,print-implicit-plus=true,print-exponent-implicit-plus=true]{1.780201e+02} \ensuremath{\pm} \num[print-zero-exponent = true,print-exponent-implicit-plus=true]{2.9e+01} \\
 &  &  \\
\textsc{linear (recal)} & Filtering & \num[print-zero-exponent = true,print-implicit-plus=true,print-exponent-implicit-plus=true]{3.583328e+02} \ensuremath{\pm} \num[print-zero-exponent = true,print-exponent-implicit-plus=true]{2.1e+01} \\
 &  &  \\
\textsc{unscented'95} & Filtering & \num[print-zero-exponent = true,print-implicit-plus=true,print-exponent-implicit-plus=true]{5.917882e+01} \ensuremath{\pm} \num[print-zero-exponent = true,print-exponent-implicit-plus=true]{1.4e+01} \\
 &  &  \\
\textsc{unscented'95 (recal)} & Filtering & \num[print-zero-exponent = true,print-implicit-plus=true,print-exponent-implicit-plus=true]{4.307320e-01} \ensuremath{\pm} \num[print-zero-exponent = true,print-exponent-implicit-plus=true]{4.0e-02} \\
 &  &  \\
\textsc{unscented'02} & Filtering & \num[print-zero-exponent = true,print-implicit-plus=true,print-exponent-implicit-plus=true]{8.869023e+01} \ensuremath{\pm} \num[print-zero-exponent = true,print-exponent-implicit-plus=true]{1.9e+01} \\
 &  &  \\
\textsc{unscented'02 (recal)} & Filtering & \num[print-zero-exponent = true,print-implicit-plus=true,print-exponent-implicit-plus=true]{6.091593e+00} \ensuremath{\pm} \num[print-zero-exponent = true,print-exponent-implicit-plus=true]{5.6e-01} \\
\bottomrule
\end{tabular}
\end{table}

\begin{table}[htbp!]
\centering
\caption{Coverage at 95\% of the Kalman filters in closed loop in the LTI regulation problem.}
\label{tab:results_regulation_coverage95}
\begin{tabular}{lll}
\toprule
Method & Task & Value $\pm$ Monte Carlo standard error \\
\midrule
\textsc{analytic} & Filtering & \num[print-zero-exponent = true,print-implicit-plus=true,print-exponent-implicit-plus=true]{9.123362e-01} \ensuremath{\pm} \num[print-zero-exponent = true,print-exponent-implicit-plus=true]{4.7e-04} \\
 &  &  \\
\textsc{analytic (recal)} & Filtering & \num[print-zero-exponent = true,print-implicit-plus=true,print-exponent-implicit-plus=true]{9.966397e-01} \ensuremath{\pm} \num[print-zero-exponent = true,print-exponent-implicit-plus=true]{1.4e-04} \\
 &  &  \\
\textsc{linear} & Filtering & \num[print-zero-exponent = true,print-implicit-plus=true,print-exponent-implicit-plus=true]{2.035204e-03} \ensuremath{\pm} \num[print-zero-exponent = true,print-exponent-implicit-plus=true]{2.9e-04} \\
 &  &  \\
\textsc{linear (recal)} & Filtering & \num[print-zero-exponent = true,print-implicit-plus=true,print-exponent-implicit-plus=true]{4.074407e-02} \ensuremath{\pm} \num[print-zero-exponent = true,print-exponent-implicit-plus=true]{9.2e-03} \\
 &  &  \\
\textsc{unscented'95} & Filtering & \num[print-zero-exponent = true,print-implicit-plus=true,print-exponent-implicit-plus=true]{1.837184e-02} \ensuremath{\pm} \num[print-zero-exponent = true,print-exponent-implicit-plus=true]{1.5e-03} \\
 &  &  \\
\textsc{unscented'95 (recal)} & Filtering & \num[print-zero-exponent = true,print-implicit-plus=true,print-exponent-implicit-plus=true]{8.528403e-01} \ensuremath{\pm} \num[print-zero-exponent = true,print-exponent-implicit-plus=true]{2.3e-03} \\
 &  &  \\
\textsc{unscented'02} & Filtering & \num[print-zero-exponent = true,print-implicit-plus=true,print-exponent-implicit-plus=true]{6.340634e-03} \ensuremath{\pm} \num[print-zero-exponent = true,print-exponent-implicit-plus=true]{9.3e-04} \\
 &  &  \\
\textsc{unscented'02 (recal)} & Filtering & \num[print-zero-exponent = true,print-implicit-plus=true,print-exponent-implicit-plus=true]{9.464046e-01} \ensuremath{\pm} \num[print-zero-exponent = true,print-exponent-implicit-plus=true]{9.7e-03} \\
\bottomrule
\end{tabular}
\end{table}

\begin{table}[htbp!]
\centering
\caption{Volume at 95\% of the Kalman filters in closed loop in the LTI regulation problem.}
\label{tab:results_regulation_volume95}
\begin{tabular}{lll}
\toprule
Method & Task & Value $\pm$ Monte Carlo standard error \\
\midrule
\textsc{analytic} & Filtering & \num[print-zero-exponent = true,print-implicit-plus=true,print-exponent-implicit-plus=true]{4.513682e-04} \ensuremath{\pm} \num[print-zero-exponent = true,print-exponent-implicit-plus=true]{4.6e-07} \\
 &  &  \\
\textsc{analytic (recal)} & Filtering & \num[print-zero-exponent = true,print-implicit-plus=true,print-exponent-implicit-plus=true]{2.381946e-03} \ensuremath{\pm} \num[print-zero-exponent = true,print-exponent-implicit-plus=true]{8.2e-06} \\
 &  &  \\
\textsc{linear} & Filtering & \num[print-zero-exponent = true,print-implicit-plus=true,print-exponent-implicit-plus=true]{2.647376e-03} \ensuremath{\pm} \num[print-zero-exponent = true,print-exponent-implicit-plus=true]{7.9e-05} \\
 &  &  \\
\textsc{linear (recal)} & Filtering & --- \ensuremath{\pm} --- \\
 &  &  \\
\textsc{unscented'95} & Filtering & \num[print-zero-exponent = true,print-implicit-plus=true,print-exponent-implicit-plus=true]{3.470272e-03} \ensuremath{\pm} \num[print-zero-exponent = true,print-exponent-implicit-plus=true]{2.6e-04} \\
 &  &  \\
\textsc{unscented'95 (recal)} & Filtering & \num[print-zero-exponent = true,print-implicit-plus=true,print-exponent-implicit-plus=true]{5.801317e-01} \ensuremath{\pm} \num[print-zero-exponent = true,print-exponent-implicit-plus=true]{8.4e-03} \\
 &  &  \\
\textsc{unscented'02} & Filtering & \num[print-zero-exponent = true,print-implicit-plus=true,print-exponent-implicit-plus=true]{1.100528e-02} \ensuremath{\pm} \num[print-zero-exponent = true,print-exponent-implicit-plus=true]{4.3e-04} \\
 &  &  \\
\textsc{unscented'02 (recal)} & Filtering & \num[print-zero-exponent = true,print-implicit-plus=true,print-exponent-implicit-plus=true]{6.157532e+02} \ensuremath{\pm} \num[print-zero-exponent = true,print-exponent-implicit-plus=true]{6.5e+00} \\
\bottomrule
\end{tabular}
\end{table}

\begin{table}[htbp!]
\centering
\caption{Cross entropy of the Kalman filters in closed loop in the LTI regulation problem.}
\label{tab:results_regulation_cross_entropy}
\begin{tabular}{lll}
\toprule
Method & Task & Value $\pm$ Monte Carlo standard error \\
\midrule
\textsc{analytic} & Filtering & \num[print-zero-exponent = true,print-implicit-plus=true,print-exponent-implicit-plus=true]{-8.217721e+00} \ensuremath{\pm} \num[print-zero-exponent = true,print-exponent-implicit-plus=true]{8.6e-03} \\
 &  &  \\
\textsc{analytic (recal)} & Filtering & \num[print-zero-exponent = true,print-implicit-plus=true,print-exponent-implicit-plus=true]{-7.989011e+00} \ensuremath{\pm} \num[print-zero-exponent = true,print-exponent-implicit-plus=true]{2.5e-03} \\
 &  &  \\
\textsc{linear} & Filtering & \num[print-zero-exponent = true,print-implicit-plus=true,print-exponent-implicit-plus=true]{2.249131e+09} \ensuremath{\pm} \num[print-zero-exponent = true,print-exponent-implicit-plus=true]{8.5e+08} \\
 &  &  \\
\textsc{linear (recal)} & Filtering & --- \ensuremath{\pm} --- \\
 &  &  \\
\textsc{unscented'95} & Filtering & \num[print-zero-exponent = true,print-implicit-plus=true,print-exponent-implicit-plus=true]{3.491803e+07} \ensuremath{\pm} \num[print-zero-exponent = true,print-exponent-implicit-plus=true]{1.7e+07} \\
 &  &  \\
\textsc{unscented'95 (recal)} & Filtering & \num[print-zero-exponent = true,print-implicit-plus=true,print-exponent-implicit-plus=true]{9.613007e-01} \ensuremath{\pm} \num[print-zero-exponent = true,print-exponent-implicit-plus=true]{1.0e+00} \\
 &  &  \\
\textsc{unscented'02} & Filtering & \num[print-zero-exponent = true,print-implicit-plus=true,print-exponent-implicit-plus=true]{1.806193e+08} \ensuremath{\pm} \num[print-zero-exponent = true,print-exponent-implicit-plus=true]{7.9e+07} \\
 &  &  \\
\textsc{unscented'02 (recal)} & Filtering & \num[print-zero-exponent = true,print-implicit-plus=true,print-exponent-implicit-plus=true]{2.727023e+01} \ensuremath{\pm} \num[print-zero-exponent = true,print-exponent-implicit-plus=true]{2.1e+01} \\
\bottomrule
\end{tabular}
\end{table}

\clearpage
\section{Runtime}
\label{app:runtime}
Tables~\ref{tab:runtime_lorenz} and~\ref{tab:runtime_lti} report per-call wall-clock time and XLA cost-analysis metrics for the Kalman \texttt{Predict}/\texttt{Update} and RTS \texttt{Predict}/\texttt{Update} primitives of Algorithms~\ref{alg:kalman-filter}--\ref{alg:kalman-filter-recal}, evaluated in the Lorenz (\S\ref{sec:stochastic-lorenz-system}) and LTI (\S\ref{sec:lti-estimation}) estimation problems.
The Kalman and RTS \texttt{Update} rows for which recalibration has no effect are reported once; vanilla Kalman \texttt{Update} is likewise method-agnostic since it reduces to conditioning on a jointly Gaussian distribution.
Benchmark inputs are drawn from random (non-data) distributions of the correct shape and dtype, since only the compiled graph governs runtime.
Wall clock is reported as mean $\pm$ standard error over 10 repetitions on CPU (\texttt{jax\_platform\_name=cpu}, \texttt{jax\_enable\_x64=True}), each repetition sized so that the JIT-compiled call is invoked enough times to run for at least 0.2\,s.
FLOPs, bytes accessed, and transcendental-op counts come from \texttt{jax.jit(fn).lower().compile().cost\_analysis()}.

All in all, the \textsc{analytic} operations are up to 10 times slower than conventional methods such as \(\textsc{unscented'95}\), albeit faster than the FLOPs or other low-level operation counts would suggest.
Though the \textsc{analytic} filter and smoother represent a performance-accuracy tradeoff, we reiterate from \citet{neural-kalman-anonymous} that they ultimately reside in the same complexity class as \textsc{linear} (EKF).

\begin{table}[htbp!]
\centering
\caption{\label{tab:runtime_lorenz}Per-call runtime of Kalman and RTS primitives in the Lorenz state estimation problem. Wall clock reports mean $\pm$ standard error over 10 repetitions on CPU.}
\begin{tabular}{lrrrr}
\toprule
Operation & Wall clock ($\mu$s) & FLOPs & Bytes accessed & Transcendentals \\
\midrule
\multicolumn{5}{l}{\textit{Kalman \texttt{Predict}}} \\
\quad \textsc{analytic} & 2240.36 $\pm$ 39.46 & 11\,046\,598 & 3\,254\,294 & 103\,680 \\
\quad \textsc{linear} & 258.85 $\pm$ 8.13 & 270\,093 & 625\,310 & 640 \\
\quad \textsc{unscented'95} & 367.83 $\pm$ 13.85 & 472\,859 & 398\,116 & 2\,240 \\
\quad \textsc{unscented'02} & 290.36 $\pm$ 12.34 & 472\,859 & 398\,116 & 2\,240 \\
\quad \textsc{mean field} & 1194.42 $\pm$ 43.57 & 11\,087\,679 & 3\,262\,598 & 103\,680 \\
\addlinespace
\multicolumn{5}{l}{\textit{Kalman \texttt{Update} (vanilla)}} \\
 & 19.39 $\pm$ 0.57 & 24 & 449 & 0 \\
\addlinespace
\midrule
\multicolumn{5}{l}{\textit{RTS \texttt{Predict}}} \\
\quad \textsc{analytic} & 2233.02 $\pm$ 22.62 & 12\,676\,065 & 3\,567\,919 & 113\,565 \\
\quad \textsc{linear} & 102.22 $\pm$ 2.02 & 519\,798 & 748\,751 & 670 \\
\quad \textsc{unscented'95} & 111.37 $\pm$ 4.31 & 520\,803 & 429\,842 & 2\,345 \\
\quad \textsc{unscented'02} & 87.90 $\pm$ 1.44 & 520\,803 & 429\,842 & 2\,345 \\
\quad \textsc{mean field} & 720.53 $\pm$ 19.78 & 12\,721\,100 & 3\,576\,591 & 113\,565 \\
\addlinespace
\multicolumn{5}{l}{\textit{RTS \texttt{Update}}} \\
 & 26.81 $\pm$ 1.08 & 175 & 1\,889 & 0 \\
\bottomrule
\end{tabular}
\end{table}

\begin{table}[htbp!]
\centering
\caption{\label{tab:runtime_lti}Per-call runtime of Kalman and RTS primitives in the LTI state-estimation problem. Wall clock reports mean $\pm$ standard error over 10 repetitions on CPU.}
\begin{tabular}{lrrrr}
\toprule
Operation & Wall clock ($\mu$s) & FLOPs & Bytes accessed & Transcendentals \\
\midrule
\multicolumn{5}{l}{\textit{Kalman \texttt{Predict}}} \\
\quad \textsc{analytic} & 384.50 $\pm$ 6.13 & 697\,938 & 1\,438\,494 & 63\,000 \\
\quad \textsc{linear} & 47.78 $\pm$ 2.29 & 44\,007 & 66\,046 & 112 \\
\quad \textsc{unscented'95} & 120.99 $\pm$ 3.60 & 112\,632 & 78\,692 & 1\,456 \\
\quad \textsc{unscented'02} & 112.43 $\pm$ 1.84 & 112\,632 & 78\,692 & 1\,456 \\
\addlinespace
\multicolumn{5}{l}{\textit{Kalman \texttt{Update} (vanilla)}} \\
 & 20.38 $\pm$ 0.32 & 238 & 2\,705 & 0 \\
\addlinespace
\multicolumn{5}{l}{\textit{Kalman \texttt{Update} (recalibrated)}} \\
\quad \textsc{analytic} & 398.84 $\pm$ 9.86 & 691\,056 & 1\,430\,408 & 62\,888 \\
\quad \textsc{linear} & 42.93 $\pm$ 1.77 & 36\,796 & 57\,144 & 56 \\
\quad \textsc{unscented'95} & 90.63 $\pm$ 2.87 & 109\,929 & 73\,071 & 1\,456 \\
\quad \textsc{unscented'02} & 93.81 $\pm$ 5.76 & 109\,929 & 73\,071 & 1\,456 \\
\midrule
\multicolumn{5}{l}{\textit{RTS \texttt{Predict}}} \\
\quad \textsc{analytic} & 29.60 $\pm$ 0.63 & 8\,371 & 16\,759 & 0 \\
\quad \textsc{linear} & 36.44 $\pm$ 0.57 & 10\,197 & 19\,927 & 0 \\
\quad \textsc{unscented'95} & 52.87 $\pm$ 0.65 & 14\,454 & 23\,366 & 0 \\
\quad \textsc{unscented'02} & 53.66 $\pm$ 1.32 & 14\,454 & 23\,366 & 0 \\
\addlinespace
\multicolumn{5}{l}{\textit{RTS \texttt{Update}}} \\
 & 26.06 $\pm$ 0.60 & 683 & 4\,961 & 0 \\
\bottomrule
\end{tabular}
\end{table}

\end{document}